\documentclass[journal,10pt]{IEEEtran}
\usepackage[utf8]{inputenc}
\usepackage{amsmath,amsthm,amssymb,algorithm,tikz}
\usepackage{algpseudocode}
\usepackage{enumitem,epsfig}
\usepackage{adjustbox}
\usepackage{graphicx}
\usepackage{rotating}
\usepackage{multirow}
\usepackage{subcaption}
\usepackage{multicol}
\usepackage{float}
\usepackage{hyperref}
\usepackage{wrapfig}
\usepackage{xcolor}

\usetikzlibrary{positioning}
\usetikzlibrary{shapes,arrows}

\title{Multi-Branch Tensor Network Structure for Tensor-Train Discriminant Analysis}
\author{\textit{Seyyid Emre Sofuoglu, Selin Aviyente}\thanks{\noindent This work was in part supported by NSF CCF-1615489.}
\thanks{\noindent S. E. Sofuoglu and S. Aviyente are with Electrical and Computer Engineering, Michigan State University, East Lansing, MI, 48824, USA.
E-mails: sofuoglu@msu.edu,
aviyente@egr.msu.edu}}

\DeclareMathOperator*{\argmin}{argmin}

\theoremstyle{definition}
\newtheorem{lemma}{Lemma}

\newcommand{\vast}{\bBigg@{4}}
\newcommand{\Biggg}{\bBigg@{5}}

\def\tdt{\times\dots\times}

\def\x{\mathbf{x}}

\def\L{\mathbf{L}}
\def\R{\mathbf{R}}
\def\A{\mathcal{A}}
\def\B{\mathcal{B}}
\def\C{\mathcal{C}}
\def\D{\mathcal{D}}
\def\M{\mathcal{M}}
\def\T{\mathbf{T}}

\def\U{\mathcal{U}}
\def\S{\mathcal{S}}

\def\V{\mathbf{V}}
\def\X{\mathcal{X}}
\def\Y{\mathcal{Y}}

\def\ldist{3}
\def\rdist{3}
\def\bc{1}

\begin{document}

\maketitle
\thispagestyle{empty}
\pagestyle{empty}
\begin{abstract}
    Higher-order data with high dimensionality arise in a diverse set of application areas such as computer vision, video analytics and medical imaging. Tensors provide a natural tool for representing these types of data. Two major challenges that confound current tensor based supervised learning algorithms are storage complexity and computational efficiency. In this paper, we address these problems  by employing tensor-trains, a hierarchical tensor network structure that parameterizes large-scale multidimensional data via a network of low-rank tensors. First, we introduce a supervised discriminative subspace learning approach based on the tensor-train model, referred to as tensor-train discriminant analysis (TTDA).  We then introduce a multi-branch tensor network structure for efficient implementation of TTDA. The multi-branch approach takes advantage of the flexibility of the tensor network structure by reordering the low-rank projection and core tensors to reduce both storage and computational complexity. Multi-branch implementations of TTDA are shown to achieve lower storage and computational complexity while providing improved classification performance with respect to both Tucker based and existing tensor-train based supervised learning methods. 
\end{abstract}
\providecommand{\keywords}[1]{\textbf{\textit{Index terms---}} #1}
\section{Introduction}
\noindent Most real-world data is multidimensional, i.e. it is a function of several independent variables, and typically represented by a multidimensional array of numbers. These arrays are often referred to as tensors \cite{de2000multilinear}. For instance, a color image is a third-order tensor
defined by two indices for spatial variables and one index
for color mode.  Similarly, a video comprised of color images is a
fourth-order tensor, time being the fourth dimension besides spatial and spectral.  


Recently, tensors have received attention in machine learning community, where given a collection of training tensors  $\Y \in \mathbb{R}^{I_1\times\dots\times I_N\times K\times C}$ from $C$ classes each with $K$ samples, the goal is to extract low-dimensional features for subsequent classification tasks. Vectorizing high dimensional inputs may result in poor classification performance due to overfitting when the training sample size is relatively small compared to the feature vector dimension \cite{sidiropoulos2017tensor,li2014multilinear,tao2005supervised,tao2007general,yan2006multilinear}. For this reason, a variety of supervised tensor learning methods for feature extraction, selection, regression and classification have been proposed \cite{lu2008mpca,tao2007general,kotsia2012higher,hao2013linear,li2014multilinear,guo2016support}. Most of the existing work has utilized Tucker decomposition. However, for larger tensors, Tucker representation can be exponential in storage requirements \cite{wang2019principal,chaghazardi2017sample}.

 In order to address this issue of exponential storage requirements and high computational complexity, in this paper, we introduce a supervised subspace learning approach based on the tensor-train (TT) structure. In particular, we present a discriminant subspace learning approach using the TT model, namely the Tensor-Train Discriminant Analysis (TTDA). The proposed approach is based on linear discriminant analysis (LDA) and learns a tensor-train subspace  (TT-subspace) \cite{wang2018tensor,wang2019principal} that maximizes the linear discriminant function. Although this approach  provides an efficient structure for storing the learnt subspaces, it is computationally prohibitive. For this reason, we propose a multi-branch tensor network structure and develop computationally efficient, low storage complexity implementations of TTDA. 

\subsection{Related Work}
The proposed work builds on two fundamental lines of research: 1) Linear supervised learning methods for tensors and 2) Tensor-Train subspace learning methods. In the area of supervised tensor learning, methods to learn discriminant subspaces from a set of labelled training examples have been proposed. These include extensions of Linear Discriminant Analysis (LDA) to Multilinear Discriminant Analysis (MDA) for face and gait recognition \cite{yan2006multilinear,tao2007general,li2014multilinear}; Discriminant Non-negative Tensor Factorization (DNTF) \cite{zafeiriou2009discriminant}; Supervised Tensor Learning (STL) where one projection vector along each mode of a tensor is learnt \cite{tao2005supervised,he2014dusk}. More recently, the linear regression model has been extended to tensors to learn multilinear mappings from a tensorial input space to a continuous output space \cite{guo2011tensor,yu2016learning}. Finally, a framework for tensor-based linear large margin classification was formulated as Support Tensor Machines (STMs), in which the parameters defining the separating hyperplane form a tensor \cite{kotsia2012higher,hao2013linear,guo2016support}. However, almost all of these methods are based on  Tucker decomposition. For large tensors, these representations are computationally expensive and their storage requirements grow exponentially \cite{cichocki2017tensor,chaghazardi2017sample}.

In \cite{holtz2012manifolds,cichocki2016tensor}, it was shown that tensor-train representation can address these shortcomings.  Tensor networks are factorizations of very large tensors into networks of smaller tensors with applications in applied mathematics, physics and machine learning \cite{orus2014practical}.  The matrix product state (MPS) or tensor-train is one of the best understood tensor networks for which  efficient algorithms have been developed \cite{oseledets2011tensor,oseledets2010approximation}. TT is a special case of a tensor network where a tensor with $N$ indices is factorized into a chain-like product of low-rank, three-mode tensors. This model provides better compression than Tucker models, especially for higher order tensors \cite{cichocki2014era}. Even though early applications of TT decomposition focused on compression and dimensionality reduction \cite{oseledets2010approximation,oseledets2011tensor}, more recently TT has been used in machine learning applications. In \cite{bengua2017matrix}, MPS is implemented in an unsupervised manner to first compress the tensor of training samples and then the resulting lower dimensional core tensors are used as features for subsequent classification. In \cite{wang2019principal}, TT decomposition is associated with a structured subspace model, namely the tensor-train subspace. Learning this structured subspace from training data is posed as a non-convex problem referred to as TT-PCA. Once the subspaces are learnt from the training data, the resulting low-dimensional subspaces are used to project and classify the test data. \cite{wang2018tensor} extends TT-PCA to manifold learning by proposing a tensor-train neighborhood preserving embedding (TTNPE). The classification is conducted by first learning a set of tensor subspaces from the training data and then projecting the training and testing data onto the learnt subspaces. Apart from employing TT for subspace learning, recent work has also considered the use of TT in classifier design. In \cite{chen2019support}, a support tensor train machine (STTM) is introduced to replace the rank-1 weight tensor in Support Tensor Machine (STM) \cite{tao2005supervised} by a tensor-train that can approximate any tensor with a scalable number of parameters.

\subsection{Contributions of the Proposed Work}

 The contributions of the proposed work can be summarized as follows:
 \begin{itemize}
     \item This paper is the first that uses tensor-train decomposition to formulate LDA for supervised subspace learning. 
     Unlike recent work on TT-subspace learning \cite{bengua2017efficient,bengua2017matrix,wang2019principal,wang2018tensor} which focuses on dimensionality reduction for feature extraction, the proposed work learns discriminative TT-subspaces and uses them to extract features that will optimize the linear discriminant function.
     \item A computationally efficient way to implement tensor-train decomposition is presented. The proposed multi-branch structure is akin to a hybrid between tensor-train and Tucker decompositions using the flexibility of tensor networks.  This structure is not limited to LDA as it can also be utilized within other subspace learning tasks, e.g. PCA. A convergence analysis for the proposed algorithm to solve the resulting non-convex optimization problem is also provided.
     \item A theoretical analysis of storage and computational complexity of this new framework is presented. A method to find the optimal implementation of the multi-branch TT model given the dimensions of the input tensor is also given.
     \item The proposed method provides higher classification accuracy at a reduced storage complexity and reduces the computational complexity by a factor of $10^2$ especially at high compression ratios compared to Tucker based supervised learning methods. Moreover, the proposed method is able to learn more discriminative subspaces from a small number of training samples compared to MDA.
 \end{itemize}

The rest of the paper is organized as follows. In Section \ref{sec:back}, we provide background on tensor operations, TT and Tucker decomposition, LDA and MDA. In Section \ref{sec:ttda}, we introduce an optimization problem to learn the TT-subspace structure that maximizes the linear discriminant function. In Section \ref{sec:mbttda}, we introduce multi-branch implementations of TTDA to address the issue of high computational complexity. In Section \ref{sec:compC}, we provide an analysis of storage cost, computational complexity and convergence for the proposed algorithms. We also provide a procedure to determine the optimal TT structure for minimizing storage complexity. In Section \ref{sec:Exp}, we compare the proposed methods with state-of-the-art tensor based discriminant analysis and subspace learning methods for classification applications.

\section{Background}
\label{sec:back}
\noindent Let $\Y \in \mathbb{R}^{I_1\times\dots\times I_N\times K\times C}$ be the collection of samples of training tensors. For a given $\Y$ with $C$ classes and $K$ samples per class, define $\Y_c^k\in \mathbb{R}^{I_1 \times I_2 \times\dots\times I_N }$ as the sample tensors where $c\in\{1,\dots,C\}$ is the class index and $k\in\{1,\dots,K\}$ is the sample index.
\subsection{Notation}
\noindent \textbf{Definition 1.} (Vectorization, Matricization and Reshaping) $\V(.)$ is a vectorization operator such that $\V(\Y_c^k)\in \mathbb{R}^{I_1I_2\dots I_N\times1}$. $\mathbf{T}_n(.)$ is a tensor-to-matrix reshaping operator defined as $\mathbf{T}_n(\Y_c^k)\in \mathbb{R}^{I_1\dots I_n\times I_{n+1}\dots I_N}$ and the inverse operator is denoted as $\T_n^{-1}(.)$.

\noindent \textbf{Definition 2.} (Left and right unfolding) The left unfolding operator creates a matrix from a tensor by taking all modes except the last mode as row indices and the last mode as column indices, i.e. $\mathbf{L}(\Y_c^k) \in \mathbb{R}^{I_1I_2\dots I_{N-1}\times I_N}$ which is equivalent to $\T_{N-1}(\Y_c^k)$. Right unfolding transforms a tensor to a matrix by taking all the first mode fibers as column vectors, i.e. $\mathbf{R}(\Y_c^k) \in \mathbb{R}^{I_1\times I_2I_3\dots I_N}$ which is equivalent to $\T_1(\Y_c^k)$. The inverse of these operators are denoted as $\L^{-1}(.)$ and $\R^{-1}(.)$, respectively.

\noindent \textbf{Definition 3.} (Tensor trace) Tensor trace is applied on matrix slices of a tensor and contracts them to scalars. Let $\A \in \mathbb{R}^{I_1\times I_2\times \dots \times I_N}$ with $I_{k'}=I_k$, then trace operation on modes $k'$ and $k$ is defined as:
\begin{gather}
    \D=tr_{k'}^k(\A)=\sum_{\substack{i_{k'}=i_k=1}}^{I_{k}}\A(:,\dots, i_{k'},:,\ldots, i_k,:,\dots, :),\nonumber
\end{gather}
where $\D \in \mathbb{R}^{I_1\tdt I_{k'-1} \times I_{k'+1} \tdt I_{k-1} \times I_{k+1}\tdt I_N}$ is a $N-2$-mode tensor.

\noindent \textbf{Definition 4.} (Tensor Merging Product) Tensor merging product connects two tensors along some given sets of modes. For two tensors $\A\in \mathbb{R}^{I_1\times I_2\times\dots\times I_N}$ and $\B\in \mathbb{R}^{J_1\times J_2\times\dots\times J_M}$ where $I_n=J_m$ and $I_{n+1}=J_{m-1}$ for some $n$ and $m$, tensor merging product is given by \cite{cichocki2017tensor}:
\begin{equation}
\C=\A\times_{n,n+1}^{m,m-1}\B. \nonumber
\end{equation}
$\C\in\mathbb{R}^{I_1\times\dots \times I_{n-1}\times I_{n+2}\times \dots \times I_N\times J_1\times\dots\times J_{m-2}\times J_{m+1}\times\dots\times J_M}$ is a $(N+M-4)$-mode tensor that is calculated as:
\small
\begin{gather}
\C(i_1,\dots , i_{n-1}, i_{n+2}, \dots , i_N, j_1,\dots, j_{m-2}, j_{m+1},\dots, j_M)= \nonumber \\ 
\sum_{t_1=1}^{I_n}\sum_{t_2=1}^{J_{m-1}}\big[\A(i_1,\dots,i_{n-1},i_n=t_1,i_{n+1}=t_2,i_{n+1},\dots,i_N) \nonumber \\ \B(j_1,\dots,j_{m-2},j_{m-1}=t_2,j_m=t_1,j_{m+1},\dots,j_M)\big].\nonumber
\end{gather}
\normalsize
A graphical representation of tensors $\A$ and $\B$ and the tensor merging product defined above is given in Fig. \ref{fig:tmp}. 

A special case of the tensor merging product can be considered for the case where $I_n=J_m$ for all  $n,m\in \{1,\dots,N-1\}, M\geq N$. In this case, the tensor merging product across the first $N-1$ modes is defined as:
\begin{gather}
    \C'=\A\times_{1,\dots,N-1}^{1,\dots,N-1}\B,
    \label{eq:cprime}
\end{gather}
where $\C' \in \mathbb{R}^{I_N\times J_N \times\dots\times J_M}$. This can equivalently be written as:
\begin{gather}
    \R(\C')=\L(\A)^\top\T_{N-1}(\B),
    \label{eq:cprime2}
\end{gather}
where $\R(\C')\in \mathbb{R}^{I_N\times\prod_{m=N}^MJ_m}$.

\def\ab{.5}
    \tikzset{
      net node/.style = {circle, minimum width=2*\ab cm, inner sep=0pt, outer sep=0pt, outer color=gray!50!cyan, inner color=cyan},
      net connect/.style = {line width=1pt, draw=blue!50!cyan!25!black},
      net thick connect/.style = {net connect, line width=2.5pt},
	  second node/.style = {circle, minimum width=2*\ab cm, inner sep=0pt, outer sep=0pt, outer color=green!25!gray!40!yellow, inner color=green!25!gray!40!yellow},
      second connect/.style = {line width=1pt, draw=red!60!gray},
      second thick connect/.style = {net connect, line width=2.5pt}
    }

\begin{figure}
        \begin{subfigure}[b]{.45\columnwidth}
        \centering
        \scalebox{0.7}{
            \begin{tikzpicture}
                \foreach \i in {1,...,2}{
                    \path (225+\i*45:1.5) node (b\i) {$I_\i$};}
                  \path (225:1.5) node (b3) {$I_{N}$};
                  \path (180:1.5) node (b4) {$I_{N-1}$};
                  \path (90:1.5) node (b5) {$\dots$};
                  
                  \path (0:0) node (b6) [second node] {$\A$};
                  \node [rotate=120] at (45: 0.7) {$\dots$};
                  \node [rotate=60] at (135: 0.7) {$\dots$};
                  
                  \foreach \i in {1,...,5}{
                  \path [second connect] (b\i) -- (b6);;}
            \end{tikzpicture}
        }
        \caption{}
        \end{subfigure}
        \begin{subfigure}[b]{.45\columnwidth}
        \centering
        \scalebox{.7}{%
            \begin{tikzpicture}
                \foreach \i in {1,...,2}{
                    \path (225+\i*45:1.5) node (b\i) {$J_\i$};}
                  \path (225:1.5) node (b3) {$J_{M}$};
                  \path (90:1.5) node (b4) {$\dots$};
                  
                  \path (0:0) node (b6) [second node] {$\B$};
                  \node [rotate=120] at (45: 0.7) {$\dots$};
                  \node [rotate=60] at (135: 0.7) {$\dots$};
                  
                  \foreach \i in {1,...,4}{
                  \path [second connect] (b\i) -- (b6);;}
            \end{tikzpicture}
            
         }
         \caption{}
        \end{subfigure}
        
        \begin{subfigure}[b]{.98\columnwidth}
        \centering
        \scalebox{.7}{%
            \begin{tikzpicture}
                  \path (270:1.5) node (a1) {$I_{N}$};
                  \path (225:1.5) node (a3) {$I_{1}$};
                  \path (180:1.5) node (a4) {$I_{2}$};
                  \path (90:1.5) node (a5) {$\dots$};
                  
                  \path (0:0) node (a6) [second node] {$\A$};
                  \node [rotate=120] at (60: 0.8) {$\dots$};
                  \node [rotate=-120] at (-60: 0.8) {$\dots$};
                  \node [rotate=60] at (135: 0.7) {$\dots$};
                  
                  \foreach \i in {1,3,4,5}{
                  \path [second connect] (a\i) -- (a6);;}
                  \tikzset{every node/.append style={xshift=3cm}}
                \foreach \i in {1}{
                    \path (270+\i*45:1.5) node (b\i) {$J_\i$};}
                  \path (270:1.5) node (b3) {$J_{M}$};
                  \path (90:1.5) node (b4) {$\dots$};
                  
                  \path (0:0) node (b6) [second node] {$\B$};
                  \node [rotate=120] at (45: 0.7) {$\dots$};
                  \node [rotate=60] at (120: 0.8) {$\dots$};
                  \node [rotate=-60] at (-120: 0.8) {$\dots$};
                  
                  \foreach \i in {1,3,4}{
                  \path [second connect] (b\i) -- (b6);;}

                  \draw [second connect] (a6) to[out=20, in=160] (b6);
                  \path [second connect] (a6) to[out=-20, in=-160] (b6); 
                  \node at (-1.5, 0.7) {$I_n,J_m$};
                  \node at (-1.5, -0.7) {$I_{n+1},J_{m-1}$};
                  
                  \path (0:1.5) node {\bf\large{=}};
                  
                  \foreach \i in {1}
                    \path (260+\i*20:1.5) node (c\i) [right=3.5cm] {$I_{\i}$};
                  \path (260:1.5) node (c2) [right=3.5cm] {$J_{M}$};
                  \path (-10:1.5) node (c3) [right=3.5cm] {$I_{n-1}$};
                  \path (10:1.5) node (c4) [right=3.5cm] {$I_{n+2}$};
                  \path (80:1.5) node (c5) [right=3.5cm] {$I_N$};
                  \path (100:1.5) node (c6) [right=3.5cm] {$J_1$};
                  \path (170:1.5) node (c7) [right=3.5cm] {$J_{m-2}$};
                  \path (190:1.5) node (c8) [right=3.5cm] {$J_{m+1}$};
                  \path (0:0) node (c10) [second node] [right=3.4cm] {$\C$};
                  \node [rotate=120] at (20: 0.7) [above right=0.05cm and 4cm] {$\dots$};
                  \node [rotate=60] at (170: 0.7) [above right=0.05cm and 4cm] {$\dots$};
                  \node [rotate=-120] at (-20: 0.5) [above right=0.05cm and 4cm] {$\dots$};
                  \node [rotate=-60] at (-160: 1) [above right=0.05cm and 4cm] {$\dots$};
                  
                  \foreach \i in {1,...,8}{
                  \path [second connect] (c\i) -- (c10);;}
              \end{tikzpicture}
        }
        \caption{}
        \end{subfigure}
    \caption{Illustration of tensors and tensor merging product using tensor network notations. Each node represents a tensor and each edge represents a mode of the tensor. (a) Tensor $\A$, (b) Tensor $\B$, (c) Tensor Merging Product between modes $(n,m)$ and $(n+1,m-1)$.}
    \label{fig:tmp}
\end{figure}
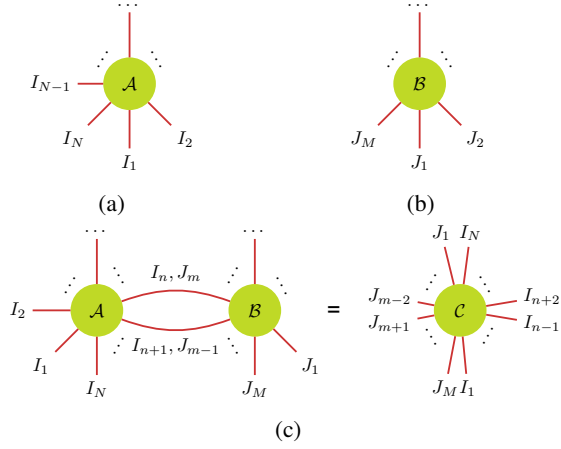%
    
\noindent \textbf{Definition 5.} (Tensor-Train Decomposition (TT)) Tensor-train decomposition represents each element of $\Y_c^k$ using a series of matrix products as:
\begin{gather}
\Y_c^k(i_1,i_2,\dots,i_N)= \nonumber \\  \U_1(1,i_1,:)\U_2(:,i_2,:) \dots \U_N(:,i_N,:)\mathbf{x}_c^k,
\label{eq:ttsc}
\end{gather}
where $\U_n\in \mathbb{R}^{R_{n-1}\times I_n \times R_n}$ are the three mode low-rank tensor factors, $R_n<I_n$ are the TT-ranks of the corresponding modes $n\in\{1,\dots,N\}$ and $\mathbf{x}_c^k\in \mathbb{R}^{R_N\times 1}$ is the projected sample vector. Using tensor merging product form, (\ref{eq:ttsc}) can be rewritten as
\begin{equation}
    \Y_c^k=\U_1\times_3^1\U_2\times_3^1\dots \times_3^1\U_N\times_3^1\mathbf{x}_c^k.
    \label{eq:proj}
\end{equation} 
A graphical representation of (\ref{eq:proj}) can be seen in Fig. \ref{fig:ttd}. If $\Y_c^k$ is vectorized, another equivalent expression for (\ref{eq:ttsc}) in terms of matrix projection is obtained as:
\begin{equation}
    \V(\Y_c^k)=\mathbf{L}(\U_1\times_3^1\U_2\times_3^1\dots \times_3^1\U_N)\mathbf{x}_c^k.\nonumber
\end{equation} 

Let $U=\L(\U_1\times_3^1\U_2\times_3^1\dots \times_3^1\U_N)$ where $U\in \mathbb{R}^{I_1I_2\dots I_N\times R_N}$. When $\L(\U_n)$s are left orthogonal, $U$ is also left orthogonal \cite{holtz2012manifolds}, i.e. $\L(\U_n)^\top\L(\U_n)=\mathbb{I}_{R_{n-1}I_n}, \forall n$ implies $U^\top U =\mathbb{I}_{I}, I=\prod_{n=1}^NI_n$ where $\mathbb{I}_{I}\in \mathbb{R}^{I\times I}$ is the identity matrix. 

\def\ab{.7}
\tikzset{
  net node/.style = {circle, minimum width=2*\ab cm, inner sep=0pt, outer sep=0pt, outer color=green!25!gray!50!cyan, inner color=green!25!gray!50!cyan},
  net connect/.style = {line width=1pt, draw=blue!50!cyan!25!black},
  net thick connect/.style = {net connect, line width=2.5pt},
  second node/.style = {circle, minimum width=2*\ab cm, inner sep=0pt, outer sep=0pt, outer color=green!25!gray!40!yellow, inner color=green!25!gray!40!yellow},
  second connect/.style = {line width=1pt, draw=red!60!gray},
  second thick connect/.style = {net connect, line width=2.5pt},
  third connect/.style = {line width=1pt, draw=green},
  third thick connect/.style = {net connect, line width=2.5pt},
  third node/.style = {circle, minimum width=2*\ab cm, inner sep=0pt, outer sep=0pt, outer color=green!25!gray!40!red, inner color=green!25!gray!40!red}
}

\begin{figure}
    \centering
    \scalebox{0.47}{
    \Large
    \begin{tikzpicture}
          \foreach \i in {1,...,2}
            \path (225+\i*45:2) node (c\i) [left=7cm] {$I_{\i}$};
          \path (225:2) node (c3) [left=7cm] {$I_{N}$};
          \path (180:2) node (c4) [left=7cm] {$I_{N-1}$};
          \path (90:2) node (c5) [left=7cm] {$\dots$};
          \path (0:0) node (c6) [second node] [left=6.8cm] {$\Y_c^k$};
          \node [rotate=120] at (45: 1) [above left=0.4cm and 7.3cm] {$\dots$};
          \node [rotate=60] at (135: 1) [above left=0.25cm and 7.3cm] {$\dots$};
          \path (180:6) node {\bf\large{=}};
          \foreach \i in {1,2}{
            \path (180:6-2*\i) node (n\i) [net node] {$\U_{\i}$};
            \path (180:6-2*\i) node (b\i) [below=2cm] {$I_{\i}$};
            }
          \path (180:0) node (n3) {$\dots$};
          \path (0:2) node (n4) [net node] {$\U_{N-1}$};
          \path (0:2) node (b4) [below=2cm] {$I_{N-1}$};
          \path (0:4) node (n5) [net node] {$\U_{N}$};
          \path (0:5.5) node (n6) [second node] {$\x_{c}^k$};
          \path (0:4) node (b5) [below=2cm] {$I_{N}$};
          \foreach \i in {9,10}{
            \path (180:16-2*\i) node (n\i) [below=2cm] {};}
          \foreach \i in {1,...,5}{
          \path [second connect] (c\i) -- (c6);;}
          \path [net connect] (n1) -- (n2) -- (n3)  -- (n4) -- (n5) -- (n6);;
          \path [second connect] (n1) --  (b1)  (n2)  -- (b2)  (n4)-- (b4) (n5) --(b5);;
          \draw [net connect] (-5.5,0) node[anchor=south]{$R_0=1$} -- (n1);;
    \end{tikzpicture}
    \normalsize
    }
    \caption{Tensor-Train Decomposition of $\Y_c^k$ using tensor merging products.}
    \label{fig:ttd}
\end{figure}
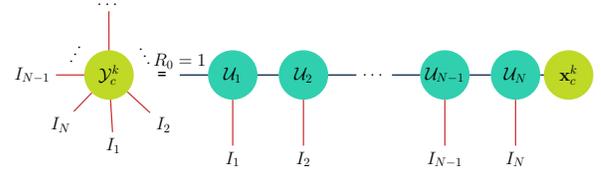

\noindent \textbf{Definition 6.} (Tucker Decomposition (TD)) If the number of modes of the projected samples, $\X_c^k$, is equal to the number of modes of the input tensors $\Y_c^k$, the TT-model becomes equivalent to Tucker decomposition. In this case, $\X_c^k$ is known as the core tensor. This is shown in Fig. \ref{fig:tckr} and given by:

\begin{equation}
    \Y_c^k=\X_c^k \times_1^2 U_1 \times_2^2 U_2 \dots \times_N^2 U_N,\nonumber
\end{equation}
where $U_n\in \mathbb{R}^{I_n\times R_n}$ and $\X_c^k\in \mathbb{R}^{R_1\times R_2\times\dots \times R_N}$.

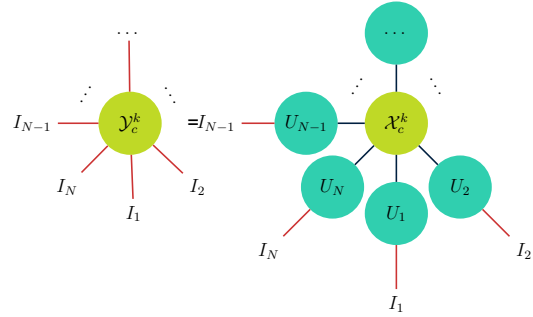
\begin{figure}
    \centering
    \scalebox{0.60}{
    \large
    \begin{tikzpicture}
        \foreach \i in {1,...,2}
            \path (225+\i*45:2) node (c\i) [left=5.5cm] {$I_{\i}$};
          \path (225:2) node (c3) [left=5.5cm] {$I_{N}$};
          \path (180:2) node (c4) [left=5.5cm] {$I_{N-1}$};
          \path (90:2) node (c5) [left=5.5cm] {$\dots$};
          \path (0:0) node (c6) [second node] [left=5.2cm] {$\Y_c^k$};
          \node [rotate=120] at (45: 1) [above left=0.4cm and 5.8cm] {$\dots$};
          \node [rotate=60] at (135: 1) [above left=0.25cm and 5.8cm] {$\dots$};
          
          \foreach \i in {1,...,5}{
          \path [second connect] (c\i) -- (c6);;}
          
          \path (180:4.5) node {\bf\large{=}};
        \foreach \i in {1,...,2}{
            \path (225+\i*45:2) node (b\i) [net node] {$U_{\i}$};
            \path (225+\i*45:4) node (n\i)  {$I_{\i}$};}
          \path (225:2) node (b3) [net node] {$U_{N}$};
          \path (225:4) node (n3)  {$I_{N}$};
          \path (180:2) node (b4) [net node] {$U_{N-1}$};
          \path (180:4) node (n4)  {$I_{N-1}$};
          \path (90:2) node (b5) [net node] {$\dots$};
          
          \path (0:0) node (b6) [second node] {$\X_c^k$};
          \node [rotate=120] at (45: 1.2) {$\dots$};
          \node [rotate=60] at (135: 1.2) {$\dots$};
          
          \foreach \i in {1,...,4}{
          \path [net connect] (b\i) -- (b6);;
          \path [second connect] (n\i) -- (b\i);;}
          \path [net connect] (b5) -- (b6);;
          
    \end{tikzpicture}
    \normalsize
    }
    \caption{Tensor network notation for Tucker decomposition.}
    \label{fig:tckr}
\end{figure}

\vspace{-1em}
\subsection{Linear Discriminant Analysis (LDA)} 

LDA for vectorized tensor data finds an orthogonal projection $U$ that maximizes the discriminability of projections \footnote{The original formulation optimizes the trace ratio. Prior work showed the equivalence of trace ratio to trace difference used in this paper \cite{fukunaga2013introduction}.}:
\begin{gather}
U=\argmin_{\hat{U}}\left[tr(\hat{U}^\top S_W\hat{U})-\lambda tr(\hat{U}^\top S_B\hat{U})\right]=\nonumber \\\argmin_{\hat{U}}tr(\hat{U}^\top (S_W-\lambda S_B)\hat{U})=\argmin_{\hat{U}}tr(\hat{U}^\top S \hat{U}),
\label{eq:LDA}
\end{gather}
where $S=S_{W}-\lambda S_{B}$, $\lambda$ is the regularization parameter that controls the trade-off between $S_W$ and $S_B$ which are within-class and between-class scatter matrices, respectively, defined as:
\begin{gather}
    S_W=\sum_{c=1}^C\sum_{k=1}^K \V(\Y_c^k-\M_c)\V(\Y_c^k-\M_c)^\top, \label{eq:LDAscat} \\
    S_B=\sum_{c=1}^C\sum_{k=1}^K \V(\M_c-\M)\V(\M_c-\M)^\top, 
    \label{eq:LDAscat2}
\end{gather}
where $\M_c=\frac{1}{K}\sum_{k=1}^K\Y_c^k$ is the mean for each class $c$ and $\M=\frac{1}{CK}\sum_{c=1}^C\sum_{k=1}^K\Y_c^k$ is the total mean of all samples. Since $U$ is an orthogonal projection, (\ref{eq:LDA}) is equivalent to minimizing the within-class scatter and maximizing the between class scatter of projections. This can be solved by the matrix $U\in \mathbb{R}^{\prod_{n=1}^NI_n\times R_N}$ whose columns are the eigenvectors of $S\in \mathbb{R}^{\prod_{n=1}^NI_n\times \prod_{n=1}^NI_n}$ corresponding to the lowest $R_N$ eigenvalues.

\subsection{Multilinear Discriminant Analysis (MDA)} 
MDA extends LDA to tensors using TD by finding a subspace $U_n\in \mathbb{R}^{I_n\times R_n}$ for each mode $n\in\{1,\dots,N\}$ that maximizes the discriminability along that mode \cite{li2014multilinear,tao2007general, yan2005discriminant}. When the number of modes $N$ is equal to 1, MDA is equivalent to LDA. In the case of MDA, within-class scatter along each mode $n \in \{1,\dots, N\}$ is defined as:
\begin{gather}
S_W^{(n)}=\sum_{c=1}^C\sum_{k=1}^{K_c} \left[(\mathcal{Y}_c^k-\M_c)\prod_{\substack{m \in \{1,\dots,N\} \\ m\neq n}}\times_m^1 {U_{m}} \right]_{(n)} \nonumber \\ {\left[(\mathcal{Y}_c^k-\M_c)\prod_{\substack{m \in \{1,\dots,N\} \\ m\neq n}}\times_m^1 {U_{m}} \right]_{(n)}}^\top.\label{eq:scatMDA}
\end{gather}
Between-class scatter $S_B^{(n)}$ is found in a similar manner. Using these definitions, each $U_n$ is found by optimizing \cite{tao2007general}:
\begin{gather}
    U_n=\argmin_{\hat{U}_n} tr(\hat{U}_n^\top (S_W^{(n)} - \lambda S_B^{(n)})\hat{U}_n).
    \label{eq:mdamoden}
\end{gather}
Different implementations of the multilinear discriminant analysis have been introduced including Discriminant Analysis with Tensor Representation (DATER), Direct Generalized Tensor Discriminant Analysis (DGTDA) and Constrained MDA (CMDA). DATER minimizes the ratio ${tr(U_n^\top S_W^{(n)}U_n)}/{tr(U_n^\top S_B^{(n)}U_n)}$ \cite{yan2005discriminant} instead of (\ref{eq:mdamoden}). Direct Generalized Tensor Discriminant Analysis (DGTDA), on the other hand, computes scatter matrices without projecting inputs on $U_m$, where $m\neq n$ and finds an optimal $U_n$\cite{li2014multilinear}. Constrained MDA (CMDA) finds the solution in an iterative fashion \cite{li2014multilinear}, where each subspace is found by fixing all other subspaces. 

\section{Tensor-Train Discriminant Analysis}
\label{sec:ttda}
When the data are higher order tensors, LDA needs to vectorize them and finds an optimal projection as shown in (\ref{eq:LDA}). This creates problems as the intrinsic structure of the data is destroyed. Even though MDA addresses this problem, it is inefficient in terms of storage complexity \cite{cichocki2017tensor, chaghazardi2017sample} as it relies on TD. Thus, we propose to solve (\ref{eq:LDA}) by constraining $U=\L(\U_1\times_3^1\U_2\times_3^1\dots \times_3^1\U_N)$ to be a TT-subspace to reduce the computational and storage complexity and obtain a solution that will preserve the inherent data structure. Consequently, the obtained $U$ will still provide discriminative features and will have a TT-subspace structure. 

The goal of TTDA is to learn left orthogonal tensor factors $\U_n \in \mathbb{R}^{R_{n-1}\times I_n \times R_{n}}, n\in \{1,\dots,N\}$ using TT-model such that the discriminability of projections $\x_c^k, \forall c, k$ is maximized. First,  $\U_n$s can be initialized using TT decomposition proposed in \cite{oseledets2011tensor}. To optimize $\U_n$s for discriminability, we need to solve (\ref{eq:LDA}) for each $\U_n$, which can be rewritten using the definition of $U$ as:
\begin{gather}
\U_n=\argmin_{\hat{\U}_n} tr\bigg[\L(\U_1\times_3^1\dots\times_3^1\hat{\U}_n\times_3^1\dots\times_3^1\U_N)^\top \nonumber \\ S \L(\U_1\times_3^1\dots\times_3^1\hat{\U}_n\times_3^1\dots\times_3^1\U_N)\bigg].
\label{eq:TTDA}
\end{gather}

Using the definitions presented in (\ref{eq:cprime}) and (\ref{eq:cprime2}), we can express (\ref{eq:TTDA})  in terms of tensor merging product:

\begin{gather}
    \U_n=\argmin_{\hat{\U}_n}  tr\bigg[(\U_1\times_3^1\dots \times_3^1\hat{\U}_n\times_3^1\dots\times_3^1\U_N) \times_{1,\dots,N}^{1,\dots,N} \S \nonumber \\ \times_{N+1,\dots,2N}^{1,\dots,N} (\U_1\times_3^1\dots\times_3^1\hat{\U}_n\times_3^1\dots\times_3^1\U_N)\bigg],
    \label{eq:ScatTT}
\end{gather}
where $\S=\T_N^{-1}(S)\in\mathbb{R}^{I_1\times\dots \times I_N\times I_1\times\dots \times I_N}$. Let $\U_{n-1}^L=\U_1\times_3^1\U_2\times_3^1\dots\times_3^1\U_{n-1}$ and $\U_{n}^R=\U_{n+1}\times_3^1\dots\times_3^1\U_N$. By rearranging the terms in (\ref{eq:ScatTT}), we can first compute all merging products and trace operations that do not involve $\U_n$ as:  
\begin{gather}
    \A_n=tr_4^8\Bigg[\U_{n-1}^L\times_{1,\dots,n-1}^{1,\dots,n-1}\bigg(\U_{n}^R\times_{1,\dots,N-n}^{n+1,\dots,N} \nonumber \\ \left(\U_{n-1}^L\times_{1,\dots,n-1}^{N+1,\dots,N+n-1}(\U_{n}^R\times_{1,\dots,N-n}^{N+n+2,\dots,2N}\S)\right)\bigg)\Bigg],
    \label{eq:a}
\end{gather}
where $\A_n \in \mathbb{R}^{R_{n-1}\times I_n\times R_{n}\times R_{n-1}\times I_n\times R_{n}}$ (refer to Fig. \ref{fig:Asupp} for a graphical representation of (\ref{eq:a})). Then, we can rewrite the optimization in terms of $\U_n$:
\begin{equation}
    \U_n=\argmin_{\hat{\U}_n}\left(\hat{\U}_n\times_{1,2,3}^{1,2,3}\left(\A_n\times_{4,5,6}^{1,2,3}\hat{\U}_n\right)\right).
    \label{eq:utau}
\end{equation}
\begin{figure}
    \centering
    \scalebox{0.5}{
    \Large
    \begin{tikzpicture}
      
      \path (0:0) node (n1) [third node] {$\S$};
      
      \path (180:\ldist) node (l1) [above=5*\bc cm] {$I_{1}$};
      \path (180:\ldist) node (l2) [above=3*\bc cm] {$I_{2}$};
      \path (180:\ldist) node (l3) [above=2*\bc cm] {$\dots$};
      \path (180:\ldist) node (l4) {$I_{n}$};
      \path (180:\ldist) node (l5) [below=2*\bc cm] {$\dots$};
      \path (180:\ldist) node (l6) [below=3*\bc cm] {$I_{N-1}$};
      \path (180:\ldist) node (l7) [below=5*\bc cm] {$I_{N}$};
      \foreach \i in {1,...,7}{
      \path [second connect] (l\i) -- (n1);}
      
      \path (0:\rdist) node (r1) [above=5*\bc cm] {$I_{1}$};
      \path (0:\rdist) node (r2) [above=3*\bc cm] {$I_{2}$};
      \path (0:\rdist) node (r3) [above=2*\bc cm] {$\dots$};
      \path (0:\rdist) node (r4) {$I_{n}$};
      \path (0:\rdist) node (r5) [below=2*\bc cm] {$\dots$};
      \path (0:\rdist) node (r6) [below=3*\bc cm] {$I_{N-1}$};
      \path (0:\rdist) node (r7) [below=5*\bc cm] {$I_{N}$};
      \foreach \i in {1,...,7}{
      \path [second connect] (r\i) -- (n1);}
      
      \path (180:2*\ldist) node (ul1) [above=2.7*\bc cm] [net node] {$\U_{n-1}^L$};
      \foreach \i in {1,...,3}{
      \path [second connect] (l\i) -- (ul1);}
      \path (180:2*\ldist) node (lr1) [above=1*\bc cm] {$R_{n-1}$};
      \path [net connect] (ul1) -- (lr1);
      
      \path (180:2*\ldist) node (ul2) [below=2.7*\bc cm] [net node] {$\U_{n}^R$};
      \foreach \i in {5,...,7}{
      \path [second connect] (l\i) -- (ul2);}
      \path (180:2*\ldist) node (lr2) [below=1*\bc cm] {$R_{n}$};
      \path (180:2*\ldist) node (lrn) [below=6*\bc cm] {$R_{N}$};
      \path [net connect] (lrn) -- (ul2) -- (lr2);
      
      \path (0:2*\ldist) node (ur1) [above=2.7*\bc cm] [net node] {$\U_{n-1}^L$};
      \foreach \i in {1,...,3}{
      \path [second connect] (r\i) -- (ur1);}
      \path (0:2*\ldist) node (rr1) [above=1*\bc cm] {$R_{n-1}$};
      \path [net connect] (ur1) -- (rr1);
      
      \path (0:2*\ldist) node (ur2) [below=2.7*\bc cm] [net node] {$\U_{n}^R$};
      \foreach \i in {5,...,7}{
      \path [second connect] (r\i) -- (ur2);}
      \path (0:2*\ldist) node (rr2) [below=1*\bc cm] {$R_{n}$};
      \path (0:2*\ldist) node (rrn) [below=6*\bc cm] {$R_{N}$};
      \path [net connect] (rrn) -- (ur2) -- (rr2);
      
      \path [third connect] (rrn) -- (lrn);
      
    \end{tikzpicture}
    \normalsize
    }
    \caption{Tensor $\A_n$ is formed by first merging $\U_n^R$, $\U_{n-1}^L$ and $\S$ and then applying trace operation across $4^{th}$ and $8^{th}$ modes of the resulting tensor. The green line at the bottom of the diagram refers to the trace operator. }
    \label{fig:Asupp}
\end{figure}
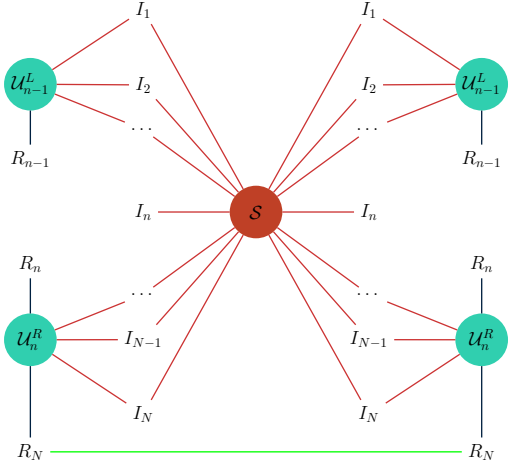
Let $A_n=\mathbf{T}_3(\A_n)\in\mathbb{R}^{R_{n-1} I_n R_{n}\times R_{n-1} I_n R_{n}}$, then (\ref{eq:utau}) can be rewritten as:
\begin{gather}
    \U_n=\argmin_{\hat{\U}_n} \V(\hat{\U}_n)^\top A_n \V(\hat{\U}_n), \nonumber \\ \L(\hat{\U}_n)^\top \L(\hat{\U}_n)=\mathbb{I}_{R_{n}}.
    \label{eq:TTDAU}
\end{gather}
This is a non-convex function due to unitary constraints and can be solved by the algorithm proposed in \cite{wen2013feasible}. The procedure described above  to find the subspaces is computationally expensive due to the complexity of finding each $\A_n$ \cite{wang2018tensor}.

When $n=N$, (\ref{eq:TTDAU}) does not apply as $\U_N^R$ is not defined and the trace operation is defined on the third mode of $\U_N$. To update $\U_N$, the following can be used:
\vspace{-2em}
\begin{gather}
    \U_N = \argmin_{\hat{\U}_N} tr\left(\hat{\U}_N\times_{1,2}^{1,2}\left(\A_N\times_{3,4}^{1,2}\hat{\U}_N\right)\right),\nonumber
\end{gather}
where $\A_N = \U_{N-1}^L\times_{1,\dots,N-1}^{1,\dots,N-1} \Big(\U_{N-1}^L\times_{1,\dots,N-1}^{N+1,\dots,2N-1}\S\Big)$. Once all of the $\U_{n}$s are obtained, they can be used to extract low-dimensional, discriminative features as $U^\top \mathbf{V}(\Y_c^k)$. The pseudocode for TTDA is given in Algorithm 1.

\algrenewcommand\algorithmicrequire{\textbf{Input:}}
\algrenewcommand\algorithmicensure{\textbf{Output:}}
\begin{algorithm}
\caption{Tensor-Train Discriminant Analysis (TTDA)}
\begin{algorithmic}[1]
\Require Input tensors $\Y_c^k \in \mathbb{R}^{I_1 \times I_2 \times \dots \times I_N }$ where $c \in \{1,\dots ,C\}$ and $k \in \{1,\dots,K\}$, initial tensor factors ${\U_n},  n \in \{1,\dots ,N\}$,  $\lambda$, $R_1,\dots,R_N$, $MaxIter$
\Ensure $\U_n, n \in \{1,\dots ,N\}$, and $\x_c^k,\quad \forall c,k$
\State $\S \gets \T_N^{-1}(S_W-\lambda S_B), \text{see eqns}. (\ref{eq:LDAscat}), (\ref{eq:LDAscat2})$.
\While{$iter<MaxIter$}
\For{$n=1$ : $N-1$}
\State Compute $\A_{n}$ using (\ref{eq:a}).
\State $\V(\U_n)\gets\hspace{-.3em}\smash{\displaystyle\argmin_{\substack{\hat{\U}_n, \\\L(\hat{\U}_n)^\top\L(\hat{\U}_n)=\mathbb{I}_{R_{n}}}}} \hspace{-.2em}\V(\hat{\U}_n)^\top \T_3(\A_n) \V(\hat{\U}_n)$.
\vspace{.7em}
\EndFor
\State $\A_{N} \gets \U_{N-1}^L\times_{1,\dots,N-1}^{1,\dots,N-1} \Big(\U_{N-1}^L\times_{1,\dots,N-1}^{N+1,\dots,2N-1}\S\Big)$
\State $\L(\U_N) \gets \hspace{-.7em}\smash{\displaystyle\argmin_{\substack{\hat{\U}_N,\\ \L(\hat{\U}_n)^\top\L(\hat{\U}_N)=\mathbb{I}_{R_{N}}}}} \hspace{-.7em}tr\Big(\L(\hat{\U}_N)^\top\T_2(\A_N)\L(\hat{\U}_N)\Big)$.
\vspace{1em}
\State $iter \gets iter+1$.
\EndWhile
\State $U=\L(\U_1\times_3^1\U_2\times_3^1\dots \times_3^1\U_N)$
\State $\x_c^k \gets U^\top \mathbf{V}(\Y_c^k)$, $\quad \forall c,k$.
\end{algorithmic}
\label{alg:TTDA}
\end{algorithm}
\vspace{-1em}
\section{Multi-Branch Tensor-Train Discriminant Analysis}
\label{sec:mbttda}
  TTDA algrorithm described above becomes computationally expensive as it requires the computation of tensor $\A_n$ through tensor merging products.  For this reason, in this section we introduce computationally efficient tensor network structures for TTDA. These new algorithms are inspired by prior work in tensor networks which considers the benefits of reshaping high-dimensional vector- and matrix-type data into tensors and then processing them using TT decomposition  \cite{cichocki2016tensor}. Several papers employed this idea to reshape matrices and vectors into tensors, known as ket augmentation and quantized TT (QTT), for better compression and higher computational efficiency \cite{wang2017efficient,wang2018tensor,cichocki2016tensor,bengua2017efficient, oseledets2010approximation,oseledets2009breaking}. 
  Inspired by this idea, we propose to tensorize the projected training samples rather than the original data in the learning framework. Using this structural approximation within TTDA formulation, we first propose to approximate 2D-LDA by TT and then generalize by increasing the number of modes (or branches) of the projected training samples.
\vspace{-.7em}
\subsection{Two-way Tensor-Train Discriminant Analysis (2WTTDA)} 
As LDA tries to find a subspace $U$ which maximizes discriminability for vector-type data, 2D-LDA tries to find two subspaces $V_1, V_2$ such that these subspaces maximize discriminability for matrix-type data \cite{ye2005two}. If one considers the matricized version of $\Y_c^k$ along mode $d$, i.e. $\T_d(\Y_c^k)\in \mathbb{R}^{\prod_{i=1}^dI_i\times \prod_{i=d+1}^NI_i}$, where $1<d<N$, the equivalent orthogonal projection can be written as:
\begin{gather}
    \T_d(\Y_c^k)=V_1X_c^kV_2^\top,
    \label{eq:2dlda}
\end{gather}
where $V_1 \!\in \! \mathbb{R}^{\prod_{i=1}^dI_i\times R_d}, V_2\! \in\! \mathbb{R}^{\prod_{i=d+1}^NI_i\times \hat{R}_d}$, $X_c^k\!\in\!\mathbb{R}^{R_d\times \hat{R}_d}$.

In TTDA, since the projections $\x_c^k$ are considered to be vectors, the subspace $U=\L(\U_1\times_3^1\U_2\times_3^1\dots \times_3^1\U_N)$ is analogous to the solution of LDA with the constraint that the subspace admits a TT model. If we consider the projections and the input samples as matrices, now we can impose a TT structure to the left and right subspaces analogous to 2D-LDA. In other words, one can find two sets of TT representations corresponding to $V_1$ and $V_2$ in (\ref{eq:2dlda}). Using this analogy, (\ref{eq:2dlda}) can be rewritten as:
\vspace{-.5em}
\begin{equation}
    \T_d(\Y_c^k)=\L(\U_1\times_3^1\dots\times_3^1 \U_d) X_c^k \R(\U_{d+1}\times_3^1\dots\times_3^1\U_N),
    \label{eq:projtw2}
\end{equation} 
which is equivalent to the following representation:
\begin{equation}
    \Y_c^k=\U_1\times_3^1\dots \times_3^1\U_d\times_3^1 X_c^k\times_2^1\U_{d+1}\times_3^1\dots\times_3^1\U_N.\nonumber
    \label{eq:projtw}
\end{equation} 
This formulation is graphically represented in Fig. \ref{fig:ttdtw} where the decomposition has two branches, thus we refer to it as Two-way Tensor-Train Decomposition (2WTT).

To maximize discriminability using 2WTT, an optimization scheme that alternates between the two sets of TT-subspaces can be utilized. When forming the scatter matrices for a set, projections of the data to the other set can be used instead of the full data which is similar to (\ref{eq:scatMDA}). This will reduce computational complexity as the cost of computing scatter matrices and the number of matrix multiplications to find $\A_n$ in (\ref{eq:a}) will decrease. We propose the procedure given in Algorithm \ref{alg:2WTTDA} to implement this approach and refer to it as Two-way Tensor-Train Discriminant Analysis (2WTTDA) as illustrated in Fig. \ref{fig:alg2w}. To determine the value of $d$ in (\ref{eq:projtw2}), we use a center of mass approach and find the $d$ that minimizes $|\prod_{i=1}^d I_i-\prod_{j=d+1}^N I_{j}|$. In this manner, the problem can be separated into two parts which have similar computational complexities. 

\algrenewcommand\algorithmicrequire{\textbf{Input:}}
\algrenewcommand\algorithmicensure{\textbf{Output:}}
\begin{algorithm}
\caption{Two-Way Tensor-Train Discriminant Analysis (2WTTDA)}
\begin{algorithmic}[1]
\Require Input tensors $\Y_c^k \in \mathbb{R}^{I_1 \times I_2 \times \dots \times I_N }$ where $c \in \{1,\dots ,C\}$ and $k \in \{1,\dots,K\}$, initial tensor factors ${\U_n},  n \in \{1,\dots ,N\}$, $d$, $\lambda$, $R_1,\dots,R_N$, $MaxIter$, $LoopIter$
\Ensure $\U_n, n \in \{1,\dots ,N\}$, and $X_c^k,\quad \forall c,k$
\While{ $iter<LoopIter$}
\State $\Y_L \gets \Y\times_{d+1,\dots,N}^{2,\dots,N-d+1}(\U_{d+1}\times_3^1\dots\times_3^1\U_N)$.
\State $[\U_i]\gets TTDA(\Y_L,\lambda, R_i, MaxIter) \forall i\in \{1,\dots,d\}$.
\State $\Y_R \gets \Y\times_{1,\dots,d}^{2,\dots,d+1}(\U_{1}\times_3^1\dots\times_3^1\U_d)$.
\State $[\U_i]\gets TTDA(\Y_R,\lambda, R_i, MaxIter) \forall i\in \{d+1,\dots,N\}$.
\State $iter=iter+1$.
\EndWhile
\State $\X_c^k \gets \L(\U_1\times_3^1\dots\times_3^1\U_d)^\top\T_d(\Y_c^k)\R(\U_{d+1}\times_3^1\dots\times_3^1\U_N)^\top$.
\end{algorithmic}
\label{alg:2WTTDA}
\end{algorithm}

\vspace{-2em}
\subsection{Three-way Tensor-Train Discriminant Analysis (3WTTDA)} 

Elaborating on the idea of 2WTTDA, one can increase the number of modes of the projected samples which will increase the number of tensor factor sets, or the number of subspaces to be approximated using TT structure. For example, one may choose the number of modes of the projections as three, i.e. $\X_c^k \in \mathbb{R}^{R_{d_1}\times R_{d_2} \times \hat{R}_{d_2}}$, where $1<d_1<d_2<N$. This model, named as Three-way Tensor-Train Decomposition (3WTT), is given in (\ref{eq:ttscthw}) and represented graphically in Fig. \ref{fig:ttdthw}. 
\vspace{-.5em}
\begin{gather}
\Y_c^k= \Bigg( \bigg(\X_c^k \times_3^{N-d_2+2} \big(\U_{d_2+1} \times_3^1 \dots\times_3^1 \U_N\big)\bigg) \times_2^{d_2-d_1+2} \nonumber \\ \big(\U_{d_1+1}\times_3^1\dots \times_3^1 \U_{d_2}\big)\Bigg)\times_1^{d_1+2} \big(\U_{1}\times_3^1\dots \times_3^1\U_{d_1} \big).
\label{eq:ttscthw}
\end{gather}

To maximize discriminability using 3WTT, one can utilize an iterative approach as in Algorithm \ref{alg:2WTTDA}, where inputs are projected on all tensor factor sets except the set to be optimized, then TTDA is applied to the projections. \textcolor{black}{The flowchart for the corresponding algorithm is illustrated in Fig. \ref{fig:alg3w}}. This procedure can be repeated until a convergence criterion is met or a number of iterations is reached. The values of $d_{1}$ and $d_{2}$ are calculated such that the product of dimensions corresponding to each set is as close to $(\prod_{i=1}^NI_i)^{1/3}$ as possible. It is important to note that 3WTT will only be meaningful for tensors of order three or higher. For three-mode tensors, 3WTT is equivalent to Tucker Model. When there are more than four modes, the number of branches can be increased accordingly.

\def\pleft{120}
\def\pright{60}
\def\pmiddle{90}
\tikzstyle{block} = [rectangle, draw, fill=white!80!cyan,
    text width=5em, text centered, rounded corners, minimum height=4em]
\tikzstyle{line} = [draw, very thick, color=black!50, -latex']
\tikzstyle{cloud} = [draw, ellipse, fill=olive!20, minimum height=2em]
\tikzset{
      left connect/.style = {line width=1pt, draw=blue!50!cyan!45!gray, -latex'},
      right connect/.style = {line width=1pt, draw=red!40!gray!60!orange, -latex'},
      middle connect/.style = {line width=1pt, draw=green!50!black!40!yellow, -latex'}
    }

\begin{figure*}
    \begin{subfigure}[b]{.98\columnwidth}
        \centering
        \scalebox{0.45}{
        \Large
        \begin{tikzpicture}
            \foreach \i in {1,...,2}
                \path (225+\i*45:\bc*2) node (c\i) [left=7*\bc cm] {$I_{\i}$};
              \path (225:\bc*2) node (c3) [left=7*\bc cm] {$I_{N}$};
              \path (180:\bc*2) node (c4) [left=7*\bc cm] {$I_{N-1}$};
              \path (90:\bc*2) node (c5) [left=7*\bc cm] {$\dots$};
              \path (0:\bc*0) node (c6) [second node] [left=6.8*\bc cm] {$\Y_c^k$};
              \node [rotate=120] at (45:\bc*1) [above left=0.4*\bc cm and 7.3*\bc cm] {$\dots$};
              \node [rotate=60] at (135:\bc* 1) [above left=0.25*\bc cm and 7.3*\bc cm] {$\dots$};
              \path (180:\bc*6) node {\bf\large{=}};
              \path (180:\bc*4) node (n1) [net node] {$\U_{1}$};
              \path (180:\bc*4) node (b1) [below=2*\bc cm] {$I_1$};
              \path (180:\bc*2.5) node (n2) {$\dots$};
              \path (180:\bc*1) node (n3) [net node] {$\U_{d}$};
              \path (180:\bc*1) node (b2) [below=2*\bc cm] {$I_{d}$};
              \path (0:\bc*2) node (n4) [net node] {$\U_{d+1}$};
              \path (0:\bc*2) node (b4) [below=2*\bc cm] {$I_{d+1}$};
              \path (0:\bc*3.5) node (n7) {$\dots$};
              \path (0:\bc*5) node (n5) [net node] {$\U_{N}$};
              \path (0:\bc*0.5) node (n6) [second node] {$X_{c}^k$};
              \path (0:\bc*5) node (b5) [below=2*\bc cm] {$I_{N}$};
              \foreach \i in {9,10}{
                \path (180:\bc*16-\bc*2*\i) node (n\i) [below=1.5*\bc cm] {};}
              \foreach \i in {1,...,5}{
              \path [second connect] (c\i) -- (c6);;}
              \path [net connect] (n1) -- (n2) -- (n3) -- (n6) -- (n4) -- (n7) -- (n5);;
              \path [second connect] (n1) --  (b1)  (n3)  -- (b2)  (n4)-- (b4) (n5) --(b5);;
        \end{tikzpicture}
        \normalsize
        }
        \caption{}
        \label{fig:ttdtw}
    \end{subfigure}
    \begin{subfigure}[b]{.98\columnwidth}
        \scalebox{0.5}{
        \Large
        \begin{tikzpicture}
            \foreach \i in {1,...,2}
                \path (225+\i*45:2) node (c\i) [left=7cm] {$I_{\i}$};
              \path (225:2) node (c3) [left=7cm] {$I_{N}$};
              \path (180:2) node (c4) [left=7cm] {$I_{N-1}$};
              \path (90:2) node (c5) [left=7cm] {$\dots$};
              \path (0:0) node (c6) [second node] [left=6.8cm] {$\Y_c^k$};
              \node [rotate=120] at (45: 1) [above left=0.4cm and 7.3cm] {$\dots$};
              \node [rotate=60] at (135: 1) [above left=0.25cm and 7.3cm] {$\dots$};
              \path (180:6) node {\bf\large{=}};
              
              \path (0:0.5) node (n6) [second node] {$\X_{c}^k$};
              
              \path (180:4) node (n1) [below=1cm] [net node] {$\U_{d_1+1}$};
              \path (180:4) node (b1) [below=3cm] {$I_{d_1+1}$};
              \path (180:2.5) node (n2) [below=1.5cm] {$\dots$};
              \path (180:1) node (n3) [below=1cm] [net node] {$\U_{d_2}$};
              \path (180:1) node (b2) [below=3cm] {$I_{d_2}$};
              
              \path (0:2) node (n4) [below=1cm] [net node] {$\U_{d_2+1}$};
              \path (0:2) node (b4) [below=3cm] {$I_{d_2+1}$};
              \path (0:3.5) node (n7) [below=1.5cm] {$\dots$};
              \path (0:5) node (n5) [below=1cm] [net node] {$\U_{N}$};
              \path (0:5) node (b5) [below=3cm] {$I_{N}$};
              
              \path (180:4) node (n8) [above=1cm] [net node] {$\U_{1}$};
              \path (180:4) node (b8) [above=3cm] {$I_{1}$};
              \path (180:2.5) node (n9) [above=1.5cm] {$\dots$};
              \path (180:1) node (n10) [above=1cm] [net node] {$\U_{d_1}$};
              \path (180:1) node (b10) [above=3cm] {$I_{d_1}$};
              
              \foreach \i in {1,...,5}{
              \path [second connect] (c\i) -- (c6);;}
              \path [net connect] (n8) -- (n9) -- (n10) -- (n6);;
              \path [second connect] (n8) --  (b8)  (n10)  -- (b10);;
              \path [net connect] (n1) -- (n2) -- (n3) -- (n6) -- (n4) -- (n7) -- (n5);;
              \path [second connect] (n1) --  (b1)  (n3)  -- (b2)  (n4)-- (b4) (n5) --(b5);;
        \end{tikzpicture}
        \normalsize
        }
        \caption{}
        \label{fig:ttdthw}
    \end{subfigure}
    
    \begin{subfigure}[b]{.98\columnwidth}
        \centering
        \begin{tikzpicture}[scale=1, node distance = 2.5cm, auto]
            \node [cloud] (data) {$\Y, \lambda, d$};
            \draw (0, 1.5) node [block] (ttpca) {Apply TT in \cite{oseledets2011tensor}};
            \node [cloud, left of=data] (factorl) {${\U_1,\dots, \U_d}$};
            \node [cloud, right of=data] (factorr) {${\U_{d+1},\dots, \U_N}$};
            \node [block, below of = factorl] (lproj) {Project $\Y$ according to line 4.};
            \node [block, below of = data] (TTDA1) {Apply TTDA};
            \node [block, below of = factorr] (rproj) {Project $\Y$ according to line 2.};
            
            \path [line] (data) -- (ttpca);
            \path [left connect] (lproj) -- node [below, color=black] {$\Y_L$} (TTDA1);
            \path [right connect] (rproj) -- node [color=black] {$\Y_R$} (TTDA1);
            \path [left connect, dashed] (TTDA1) -- (factorr);
            \path [right connect, dashed] (TTDA1) -- (factorl);
            \path [left connect, dashed] (data) -- (lproj);
            \path [left connect, dashed] (ttpca) -- (factorl);
            \path [right connect, dashed] (data) -- (rproj);
            \path [right connect, dashed] (ttpca) -- (factorr);
            \path [left connect, dashed] (factorl) -- (lproj);
            \path [right connect, dashed] (factorr) -- (rproj);
        \end{tikzpicture}
        \caption{}
        \label{fig:alg2w}
    \end{subfigure}
    \begin{subfigure}[b]{.98\columnwidth}
        
        \begin{tikzpicture}[scale=1, node distance = 2cm, auto]
            \draw (-4, -.75) node [cloud] (data) {$\Y, \lambda, d$};
            \draw (0, 2.5) node [block] (ttpca) {Apply TT in \cite{oseledets2011tensor}};
            \draw (-2.5, 0) node [cloud] (factorl) {${\U_1,\dots, \U_{d_1}}$};
            \draw (0, .75) node [cloud] (factorm) {${\U_{d_1+1},\dots, \U_{d_2}}$};
            \draw (2.5, 0) node [cloud] (factorr) {${\U_{d_2+1},\dots, \U_N}$};
            \draw (-2, -2.5) node [block] (proj) {Project $\Y$};
            \draw (2, -2.5) node [block] (ttda) {Apply TTDA};
            
            \path [line] (data) to [out=90 , in=180] (ttpca);
            \path [left connect] (proj) to [out=-20,in=-160] node [below, color=black] {$\Y_L$} (ttda);
            \path [right connect] (proj) to [out=20,in=160] node [below, color=black] {$\Y_R$} (ttda);
            \path [middle connect] (proj) -- node [below, color=black] {$\Y_M$} (ttda);
            \path [left connect] (ttda) -- (factorl);
            \path [right connect] (ttda) -- (factorr);
            \path [middle connect] (ttda) -- (factorm);
            \path [left connect, dashed] (data) to [out=-70,in=\pleft] (proj);
            \path [line, dashed] (ttpca) -- (factorl);
            \path [right connect,dashed] (data) to [out=-10,in=\pright] (proj);
            \path [line,dashed] (ttpca) -- (factorr);
            \path [middle connect,dashed] (data) to [out=-30,in=\pmiddle] (proj);
            \path [line,dashed] (ttpca) -- (factorm);
            \path [right connect, dashed] (factorl) to [out=-40,in=\pright] (proj);
            \path [middle connect, dashed] (factorl) to [out=-90, in=\pmiddle] (proj);
            \path [left connect, dashed] (factorm) to [out=-100,in=\pleft] (proj);
            \path [right connect, dashed] (factorm) to [out=-90,in=\pright] (proj);
            \path [left connect, dashed] (factorr)  to [out=-175,in=\pleft] (proj);
            \path [middle connect, dashed] (factorr) to [out=-170,in=\pmiddle] (proj);
        \end{tikzpicture}
        \caption{}
        \label{fig:alg3w}
    \end{subfigure}
    \caption{\textcolor{black}{Illustration of the proposed methods: (a) The proposed tensor network structure for 2WTT; (b) The proposed tensor network structure for 3WTT; (c) The flow diagram for 2WTTDA (Algorithm \ref{alg:2WTTDA}); (d) The flow diagram for 3WTTDA}}
\end{figure*}
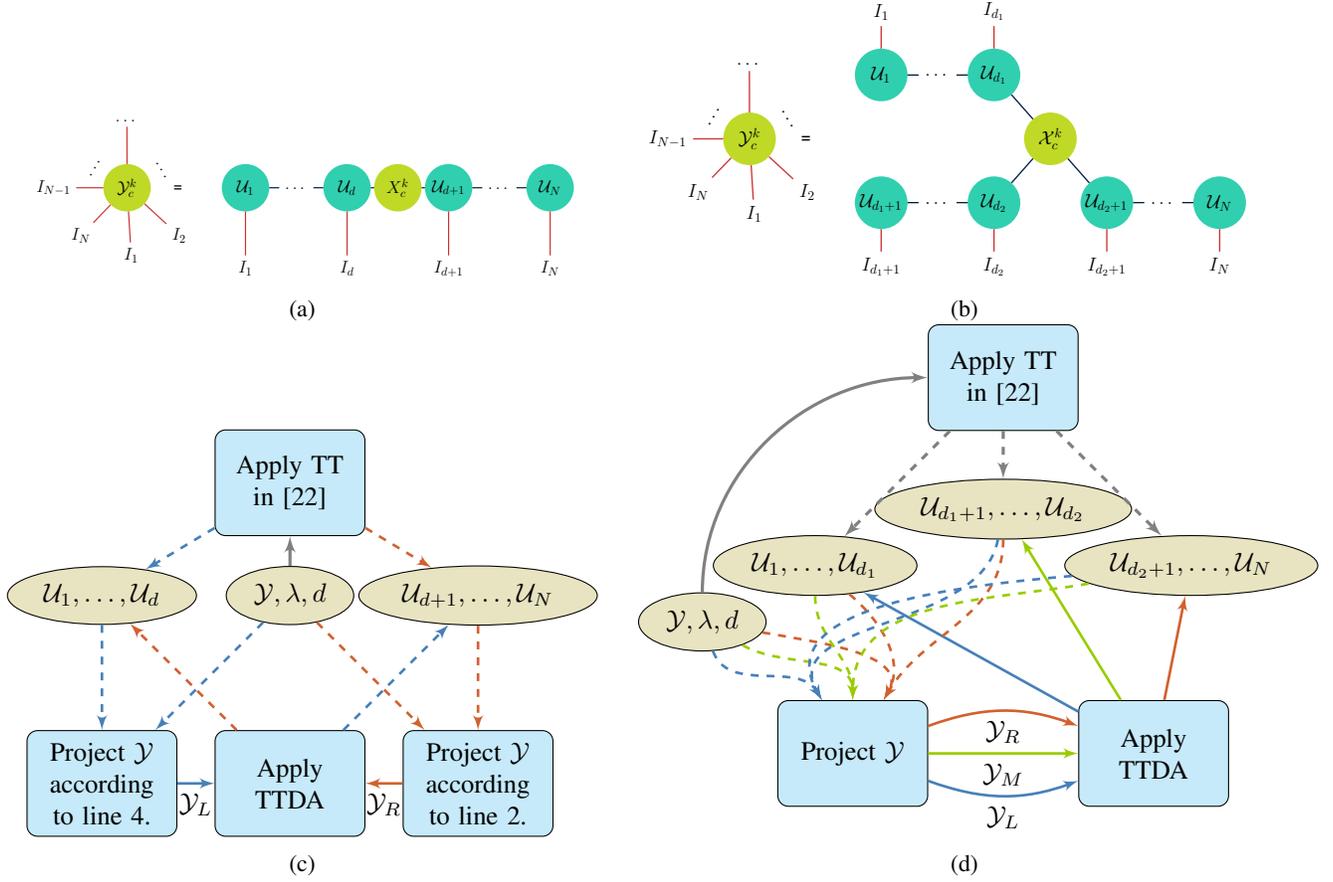

\vspace{-.7em}

\section{Analysis of Storage, Training Complexity and Convergence}
\label{sec:compC}
In this section, we derive the storage and computational complexities of the aforementioned algorithms as well as providing a convergence analysis for TTDA.
\vspace{-.5em}
\subsection{Storage Complexity} 
Let $I_n=I,  n \in \{1,2,\ldots,N\}$ and $R_l=r, l\in \{2,\dots, N-1\}$. Assuming $N$ is a multiple of both 2 and 3, 
total  storage complexities are:

\begin{itemize}
    \item $\mathcal{O}((N-1)r^2I+rI+rCK)$ for TT Decomposition, where $R_1=1, R_N=r$;
    \item $\mathcal{O}((N-2)r^2I+2rI+r^2CK)$ for Two-Way TT Decomposition, where $R_1=R_N=1$;
    \item $\mathcal{O}((N-3)r^2I+3rI+r^3CK)$ for Three-Way TT Decomposition, where $R_1=R_{d_1}=R_N=1$;
    \item $\mathcal{O}(NrI+r^NCK)$ for Tucker Decomposition, where $R_1=R_N=r$.
\end{itemize}

\begin{table}
    \centering
    \caption{Storage Complexities of Different Tensor Decomposition Structures}
    \begin{tabular}{c|c|c}
        Methods & Subspaces ($\U_{n}$s) ($\mathcal{O}(.)$)  & Projections $\X_{c}^{k}$ ($\mathcal{O}(.)$) \\
         \hline
         TT & $(N-1)r^2I+rI$ & $rCK$ \\
         2WTT &  $(N-2)r^2I+2rI$ & $r^2CK$ \\
         3WTT & $(N-3)r^2I+3rI$ & $r^3CK$ \\
         TD & $NrI$ & $r^NCK$
    \end{tabular}
    \label{tab:s_comp}
\end{table}

These results show that when the number of modes for the projected samples is increased, the storage cost increases exponentially for $\X_c^k$ while the cost of storing $\U_n$s decreases quadratically. Using the above, one can easily find the optimal number of modes for the projected samples that minimizes storage complexity. Let the number of modes of $\X_c^k$ be denoted by $f$. The storage complexity of the decomposition is then $\mathcal{O}((N-f)r^2I+f rI+r^f CK)$. The optimal storage complexity is achieved by taking the derivative of the complexity in terms of $f$ and equating it to zero. In this case, the optimal $f$ is given by \[\hat{f}= round\left(\log_r\left(\frac{r^2 I- r I}{C K \ln(r)}\right)\right),\] where $round(.)$ is an operator that rounds to the closest positive integer.
\vspace{-.5em}
\subsection{Computational Complexity}
For all of  the decompositions mentioned except for DGTDA and LDA, the $\U_n$s and $\X_c^k$ depend on each other which makes these decompositions iterative. The number of iterations will be denoted as $t_c$  and $t_t$ for CMDA and TT-based methods, respectively. For the sake of simplicity, we also define $C_s=2CK$. The total cost of finding $\U_n$s and $\X_c^k$ $\forall n, c, k$, where $r<<I$ is in the order of:
\begin{itemize}
    \item $\mathcal{O}\Big(I^N\big[(C_s+t_tr(r+N-1))I^{N}+t_tr^4(I+r^2I^{-1})\big]\Big)$ for TTDA;
    \item $\mathcal{O}\Big(rI^{N}\frac{C_s}{2}+2I^{N/2}\big[(C_s+t_tr(r+N/2-1))I^{N/2}+t_tr^4I+t_tr^6I^{-1}\big]\Big)$ for 2WTTDA;
    \item $\mathcal{O}\Big(rI^{N}\frac{C_s}{2}+3I^{N/3}\big[(C_s+t_tr(r+N/3-1))I^{N/3}+t_tr^4I+t_tr^6I^{-1}\big]\Big)$ for 3WTTDA.
\end{itemize}

\noindent If convergence criterion is met with a small number of iterations, i.e. $t_tr(r+N/f-1)<<C_s$, and $I^{N/f}>>r^6$ for all $f$, the reduced complexities are as given in Table \ref{tab:c_comp}.

\begin{table}
    \centering
    \caption{Computational complexities of various algorithms. The number of iterations to find the subspaces are denoted as $t_c$ for CMDA and $t_t$ for TT-based methods. $C_s=2CK$. ($r<<I$, $t_tr(r+N/f-1)<<C_s$, and $I^{N/f}>>r^6$)}
    \begin{tabular}{c|c}
    \hline
      Methods & Order of Complexity ($\mathcal{O}(.)$)  \\
         \hline
         LDA & $C_sI^{2N}+I^{3N}$\\
         \hline
         DGTDA & $3I^3+NC_sI^{N+1}$ \\
         \hline
         CMDA & $2t_c I^3+t_cN^2C_sI^N$\\
         \hline
         {TTDA} & $C_sI^{2N}$ \\
         \hline
         {2WTTDA} & $(r/2+2)C_sI^{N}$ \\
         \hline
         {3WTTDA} & $(rI^{N/3}/2+3)C_sI^{2N/3}$ 
    \end{tabular}
    \label{tab:c_comp}
\end{table}

We can see from Table \ref{tab:c_comp} that with increasing number of branches, TT-based methods become more efficient if the algorithm converges in a few number of iterations. This is especially the case if the ranks of tensor factors are low as this reduces the dimensionalities of the optimal solutions and the search algorithm finds a solution to (\ref{eq:TTDAU}) faster.  When this assumption holds true, the complexity is dominated by the formation of scatter matrices. Note that the ranks are assumed to be much lower than dimensionalities and number of modes is assumed to be sufficiently high. When these assumptions do not hold, the complexity of computing $\A_n$ might be dominated by terms with higher powers of $r$. This indicates that TT-based methods are more effective when the tensors have higher number of modes and when the TT-ranks of the tensor factors are low. DGTDA has an advantage over all other methods as it is not iterative and the solution for each mode is not dependent on other modes. On the other hand, the solution of DGTDA is not optimal and there are no convergence guarantees except when the ranks and initial dimensions are equal to each other, i.e. when there is no compression.


\vspace{-.5em}
\subsection{Convergence}
To analyze the convergence of TTDA, we must first establish a lower bound for the objective function of LDA, as (\ref{eq:TTDA}) is lower bounded by the objective value of (\ref{eq:LDA}).

\begin{lemma}

Given that $\lambda \in \R_+$, \textit{i.e.} a nonnegative real number, the lower bound of $tr(U^\top S_W U)-\lambda tr(U^\top S_BU)$ is achieved when $U\in \mathbb{R}^{\prod_{n=1}^N I_n\times r}$ satisfies the following two conditions, simultaneously: 
\begin{enumerate}
    \item The columns of $U$ are in the null space of $S_{W}$: $\mathbf{u}_{j} \in null(S_{W}), \forall j \in \{1,\ldots, r\}$.
    \item $\{\mathbf{u}_{1},\mathbf{u}_{2},\ldots,\mathbf{u}_{r}\}$ are the top-$r$ eigenvectors of $S_{B}$.
\end{enumerate}  
In this case, the minimum of $tr(U^\top S_W U)-\lambda tr(U^\top S_BU)=-\lambda\sum_{i=1}^{r}\sigma_{i}$, where $\sigma_{i}$s are the eigenvalues of $S_{B}$.
\end{lemma} 

\begin{proof}
Since $S_W$ is positive semi-definite,
\begin{equation*}
    0\leq \min_{U}tr(U^\top S_{W} U),
\end{equation*}
which implies that when the columns of $U$ are in the null space of $S_{W}$, i.e. $\mathbf{u}_{j} \in null(S_{W}), \forall j \in \{1,\ldots, r\}$, the minimum value will be achieved for the first part of the objective function.  

To minimize the trace difference, we need to maximize $tr(U^\top S_BU)$ which is bounded from above as:
\begin{equation*}
    \max_{U}tr(U^\top S_BU)\leq \sum_{i=1}^r\sigma_i.
\end{equation*}
$tr(U^\top S_BU)$ is maximized when the columns of $U$ are the top-$r$ eigenvectors of $S_B$. Therefore, the trace difference achieves the lower-bound when $\mathbf{u}_{j} \in null(S_{W}), \forall j \in \{1,\ldots, r\}$ and $\{\mathbf{u}_{1},\mathbf{u}_{2},\ldots,\mathbf{u}_{r}\}$ are the top-$r$ eigenvectors of $S_{B}$ and this lower-bound is equal to $-\lambda\sum_{i=1}^{r}\sigma_{i}$.
\end{proof}

As shown above, the objective function of LDA is lower bounded.  Thus, the solution to (\ref{eq:TTDA}) is also lower-bounded.

Let $f(\U_1,\U_2,\dots,\U_N)=tr\left(U^\top S U\right)$, where $U=\L(\U_1\times_3^1\dots\times_3^1\U_N)$ and $S$ is defined as in \eqref{eq:LDA}. If the function $f$ is non-increasing with each update of $\U_n$s, i.e.
\begin{gather}
    f(\U_1^t,\U_2^t,\dots,\U_n^{t-1},\dots,\U_N^{t-1})\geq \nonumber \\ f(\U_1^t,\U_2^t,\dots,\U_n^{t},\dots,\U_N^{t-1}), \qquad \forall t,n \in \{1,2,\ldots, N\},\nonumber
\end{gather}
then we can claim that Algorithm \ref{alg:TTDA} converges to a fixed point as $t\xrightarrow{} \infty$ since $f(.)$ is lower-bounded. In \cite{wen2013feasible}, an approach to regulate the step sizes in the search algorithm was introduced to guarantee global convergence. In this paper, this approach is used to update $\U_n$s. 
Thus, (\ref{eq:TTDAU}) can be optimized globally, and the objective value is non-increasing.
As Multi-Branch extensions utilize TTDA on the update of each branch, proving the convergence of TTDA is sufficient to prove the convergence of 2WTTDA or 3WTTDA.

\section{Experiments}
\label{sec:Exp}
The proposed TT based discriminant analysis methods are evaluated in terms of classification accuracy,  storage complexity, training complexity and sample size. We compared our methods\footnote{Our code is in https://github.com/mrsfgl/MBTTDA} with both linear supervised tensor learning methods including LDA, DGTDA and CMDA\cite{li2014multilinear}\footnote{https://github.com/laurafroelich/tensor\_classification} as well as other tensor-train based learning methods such as MPS \cite{bengua2017matrix}, TTNPE \cite{wang2018tensor}\footnote{https://github.com/wangwenqi1990/TTNPE.} and STTM \cite{chen2019support}\footnote{https://github.com/git2cchen/KSTTM}. The experiments were conducted on four different data sets: COIL-100, Weizmann Face, Cambridge and UCF-101. For all data sets and all methods, we evaluate the classification accuracy and training complexity with respect to storage complexity. 

In this paper, classification accuracy is evaluated using a 1-NN classifier and quantified as $N_{true}/N_{test}$, where $N_{true}$ is the number of test samples which were assigned the correct label and $N_{test}$ is the total number of test samples. Normalized storage complexity is  quantified as the ratio of the total number of elements in the learnt tensor factors ($\U_{n}, \forall n$) and projections ($\mathcal{X}_{c}^{k}, \forall c,k$) of training data, $O_s$, to the size of the original training data ($\Y_c^k, \forall c,k$): \[\frac{O_{s}}{CK\prod_{n=1}^NI_n}.\] 
Training complexity is the total runtime in seconds for learning the subspaces. All experiments were repeated 10 times with random selection of the training and test sets and average classification accuracies are reported.

The regularization parameter, $\lambda$, for each experiment was selected using a validation set composed of all of the samples in the training set and a small subset  of each class from the test set \textcolor{black}{ (10 samples for COIL-100, 5 samples for Weizmann, 1 sample for Cambridge, and 10 samples for UCF-101)}. Utilizing a leave-$s$-out approach, where $s$ is the aforementioned subset size, 5 random experiments were conducted. The optimal $\lambda$ was selected as the value that gave the best average classification accuracy among a range of values from $0.1$ to $1000$ increasing in a logarithmic scale. 
CMDA\textcolor{black}{, TTNPE and MPS} do not utilize the $\lambda$ parameter while DGTDA utilizes eigendecomposition to find $\lambda$ \cite{li2014multilinear}. \textcolor{black}{STTM has an outlier fraction parameter which was set to $0.02$ according to the original paper \cite{chen2019support}.}

\subsection{Data Sets}
\subsubsection{COIL-100}
The dataset consists of 7,200 RGB images of 100 objects of size $128\times 128$. Each object has 72 images, where each image corresponds to a different pose angle ranging from 0 to 360 degrees with increments of 5 degrees \cite{nenecolumbia}. For our experiments, we downsampled the grayscale images of all objects to $64\times 64$. Each sample image was reshaped to create a tensor of size $8\times 8\times 8\times 8$. Reshaping the inputs into higher order tensors is common practice and was studied in prior work \cite{khoromskij2011dlog, oseledets2011tensor, zhao2016tensor, cichocki2017tensor, bengua2017efficient}. 20 samples from each class were selected randomly as training data, i.e. $\Y \in \mathbb{R}^{8\times 8\times 8\times 8 \times 20 \times 100}$, and the remaining 52 samples were used for testing. 

\subsubsection{Weizmann Face Database}
The dataset includes RGB face images of size $512\times352$ belonging to 28 subjects taken from 5 viewpoints, under 3 illumination conditions, with 3 expressions \cite{weiz}. For our experiments, each image was grayscaled, and downsampled to $64\times44$. The images were then reshaped into 5-mode tensors of size $4\times4\times4\times4\times11$ as in \cite{wang2018tensor}. For each experiment, 20 samples were randomly selected to be the training data, i.e. $\Y\in \mathbb{R}^{4\times4\times4\times4\times11\times20\times28}$, and the remaining 25 samples were used in testing.

\subsubsection{Cambridge Hand-Gesture Database}
The dataset consists of 900 image sequences of 9 gesture classes, which are combinations of 3 hand shapes and 3 motions. For each class, there are 100 image sequences generated by the combinations of 5 illuminations, 10 motions and 2 subjects \cite{kim2007tensor}. Sequences consist of images of size $240\times 320$ and sequence length varies. In our experiments, we used grayscaled versions of the sequences and we downsampled all sequences to length $30$. We also included 2 subjects and 5 illuminations as the fourth mode. Thus, we have 10 samples for each of the 9 classes from which we randomly select 4 samples as the training set, i.e. $\Y\in \mathbb{R}^{30\times40\times30\times10\times4\times9}$, and the remaining 6 as test set. 

\subsubsection{UCF-101 Human Action Dataset}
UCF-101 is an action recognition dataset \cite{soomro2012ucf101}. There are 13320 videos of 101 actions, where each action category might have different number of samples. Each sample is an RGB image sequence with frame size $240\times 320\times 3$. The number of frames differs for each sample. In our experiments, we used grayscaled, downsampled frames of size $30\times 40$. From each class, we extracted 100 samples to balance the class sizes where each sample consists of 50 frames obtained by uniformly sampling each video sequence. 60 randomly selected samples from each class were used for training, i.e. $\Y \in \mathbb{R}^{30\times40\times50\times60\times101}$, and the remaining 40 samples were used for testing.
\begin{figure*}
    \centering
    \begin{subfigure}[b]{0.98\columnwidth}
        \centering
        \includegraphics[width=.98\columnwidth]{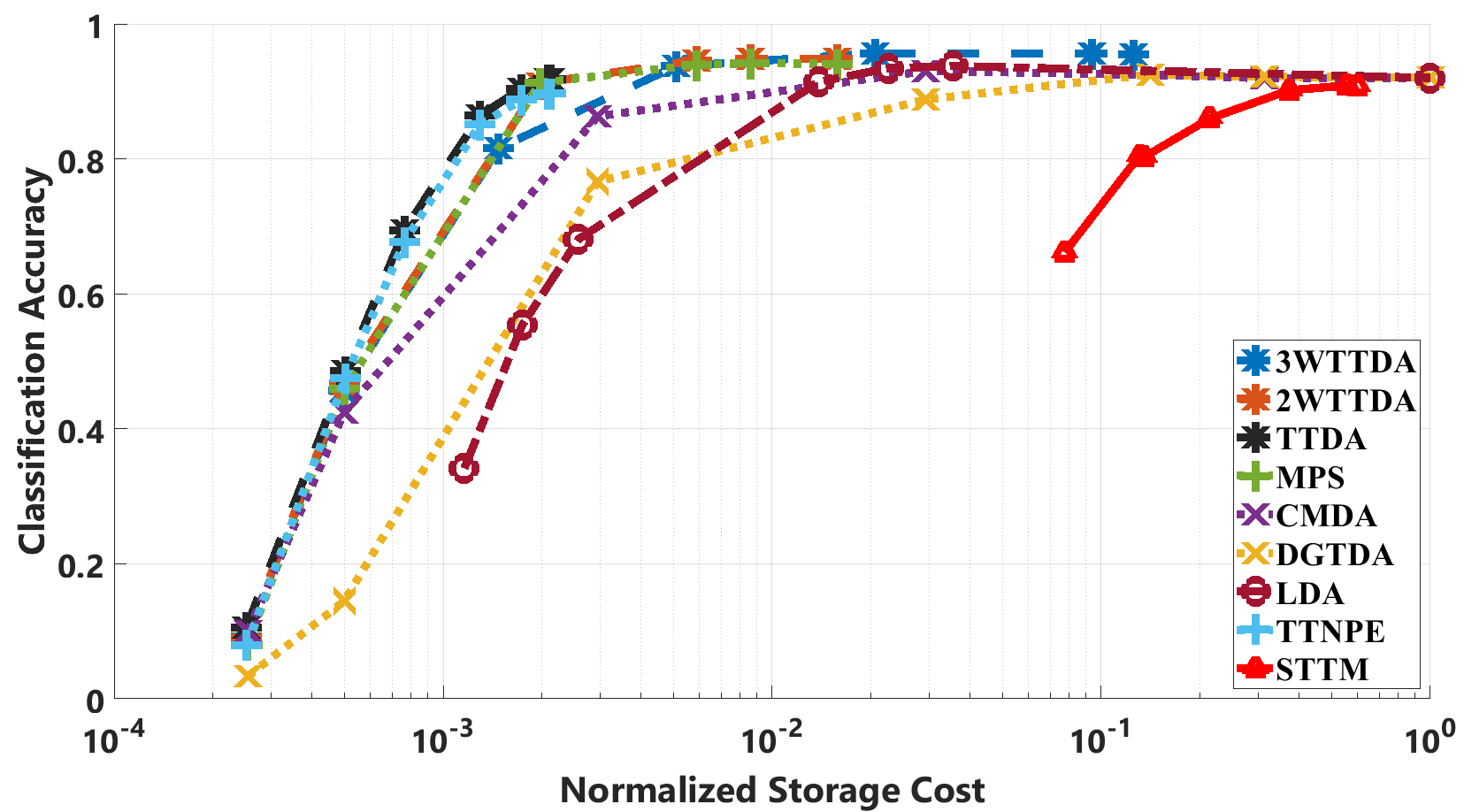}
        \caption{}
        \label{fig:clsacc}
    \end{subfigure}
    \begin{subfigure}[b]{.98\columnwidth}
        \centering
        \includegraphics[width=.98\columnwidth]{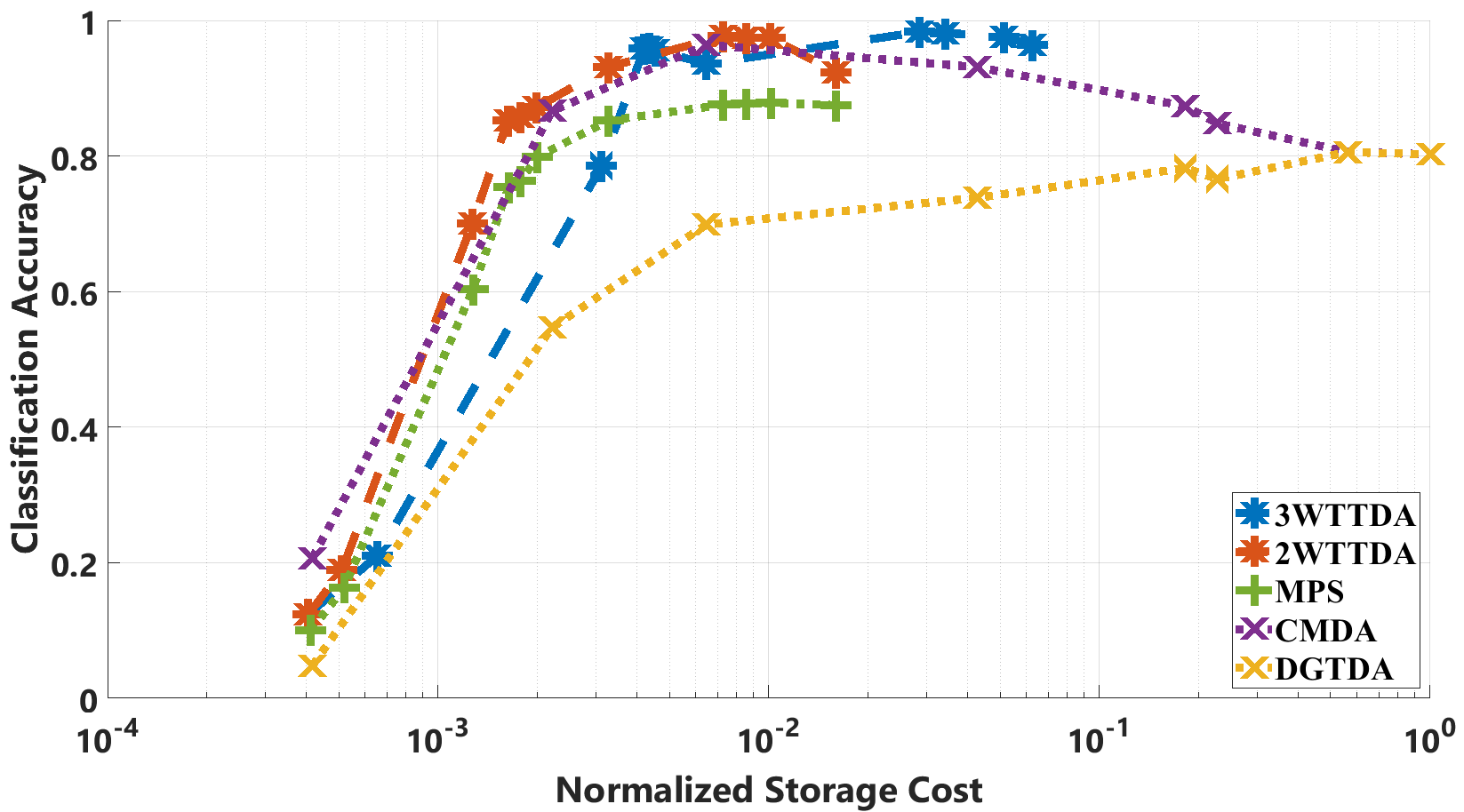}
        \caption{}
        \label{fig:weizclsacc}
    \end{subfigure}  
    
    \begin{subfigure}[b]{.98\columnwidth}
        \centering
        \includegraphics[width=.98\columnwidth]{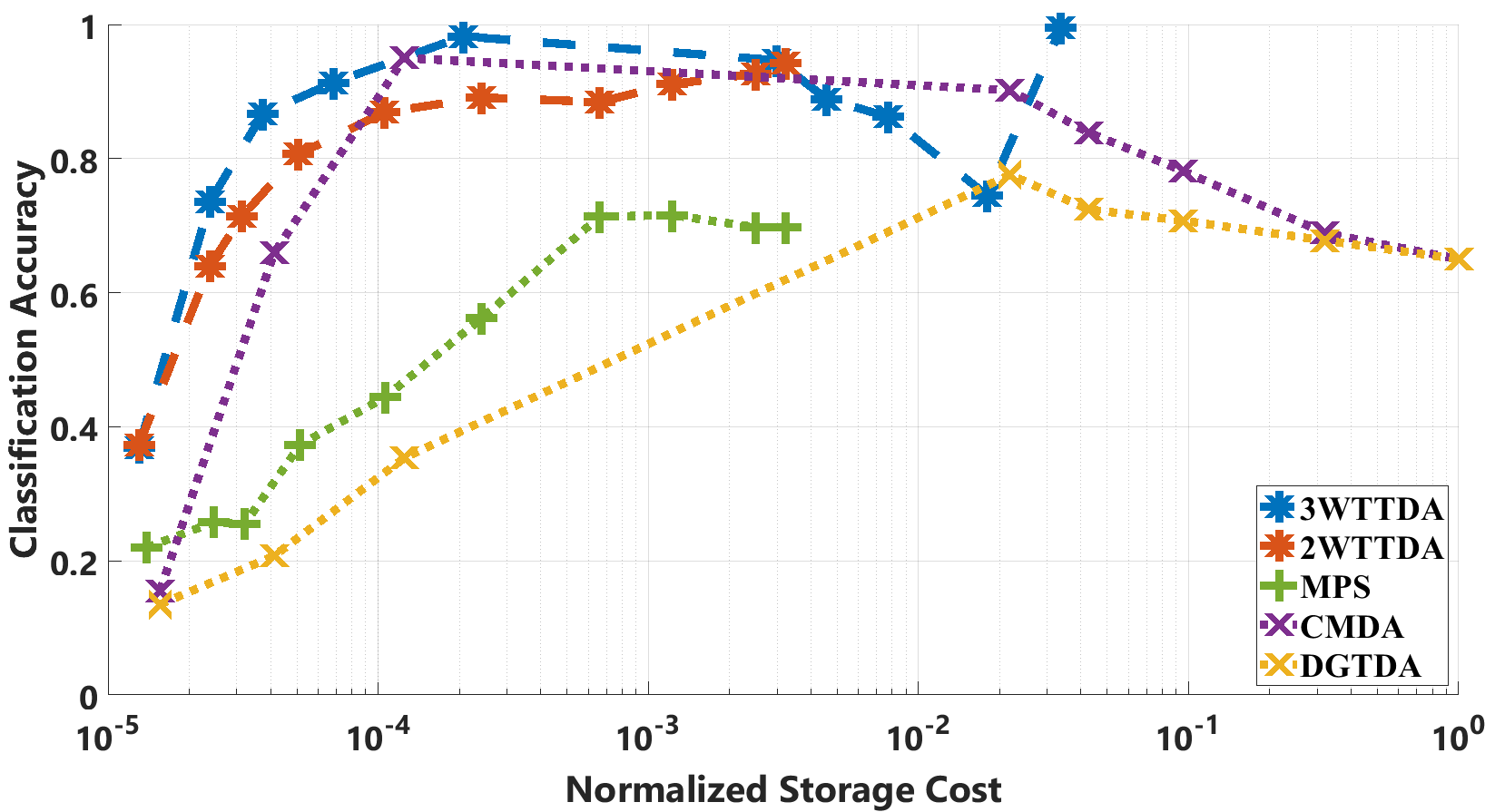}
        \caption{}
        \label{fig:cambclsacc}
    \end{subfigure}
    \begin{subfigure}[b]{.98\columnwidth}
        \centering
        \includegraphics[width=.98\columnwidth]{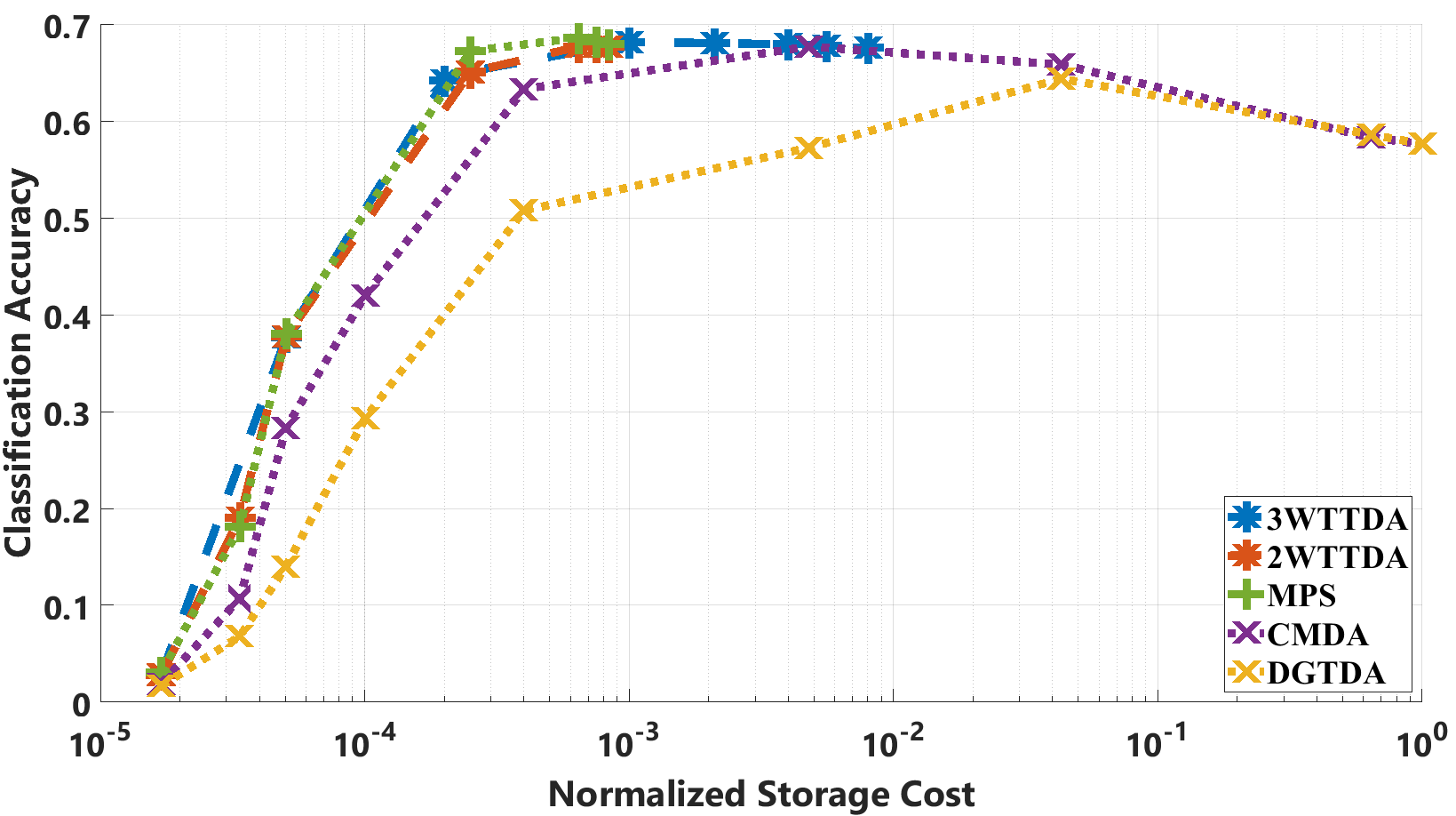}
        \caption{}
        \label{fig:ucfclsacc}
    \end{subfigure}
    \caption{\textcolor{black}{Classification accuracy vs.  Normalized storage cost of the different methods for: a) COIL-100, b) Weizmann Face, c) Cambridge Hand Gesture and d) UCF-101. All TD based methods are denoted using 'x', TT based methods are denoted using '+' and proposed methods are denoted using '*'. STTM and LDA are denoted using '$\triangle$' and 'o', respectively.}}
    \label{fig:acc_all}
\end{figure*}
\begin{figure*}
    \centering
    \begin{subfigure}[b]{.98\columnwidth}
        \centering
        \includegraphics[width=.98\columnwidth]{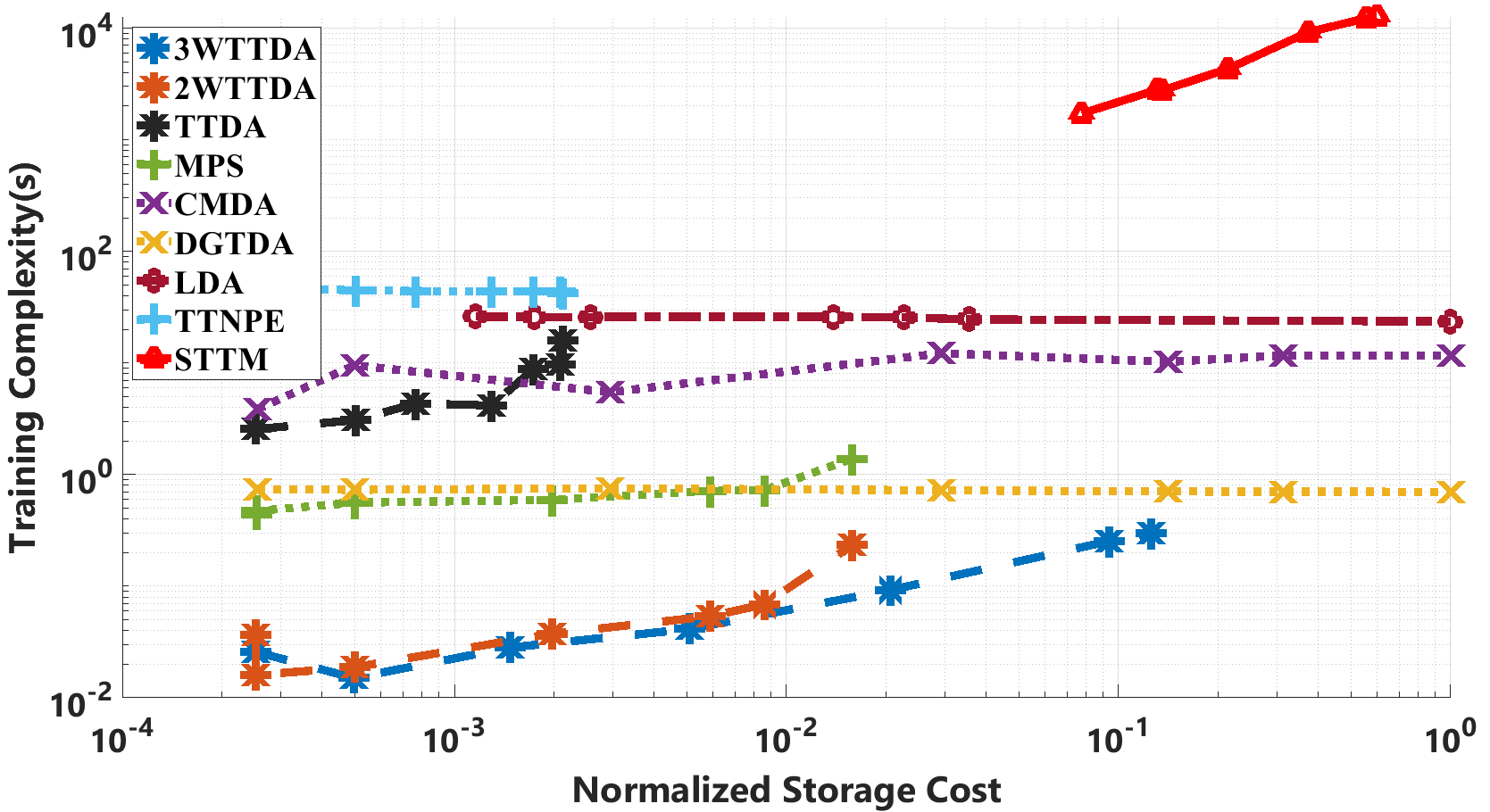}
        \caption{}
        \label{fig:subtime}
    \end{subfigure}
    \begin{subfigure}[b]{.98\columnwidth}
        \centering
        \includegraphics[width=.98\columnwidth]{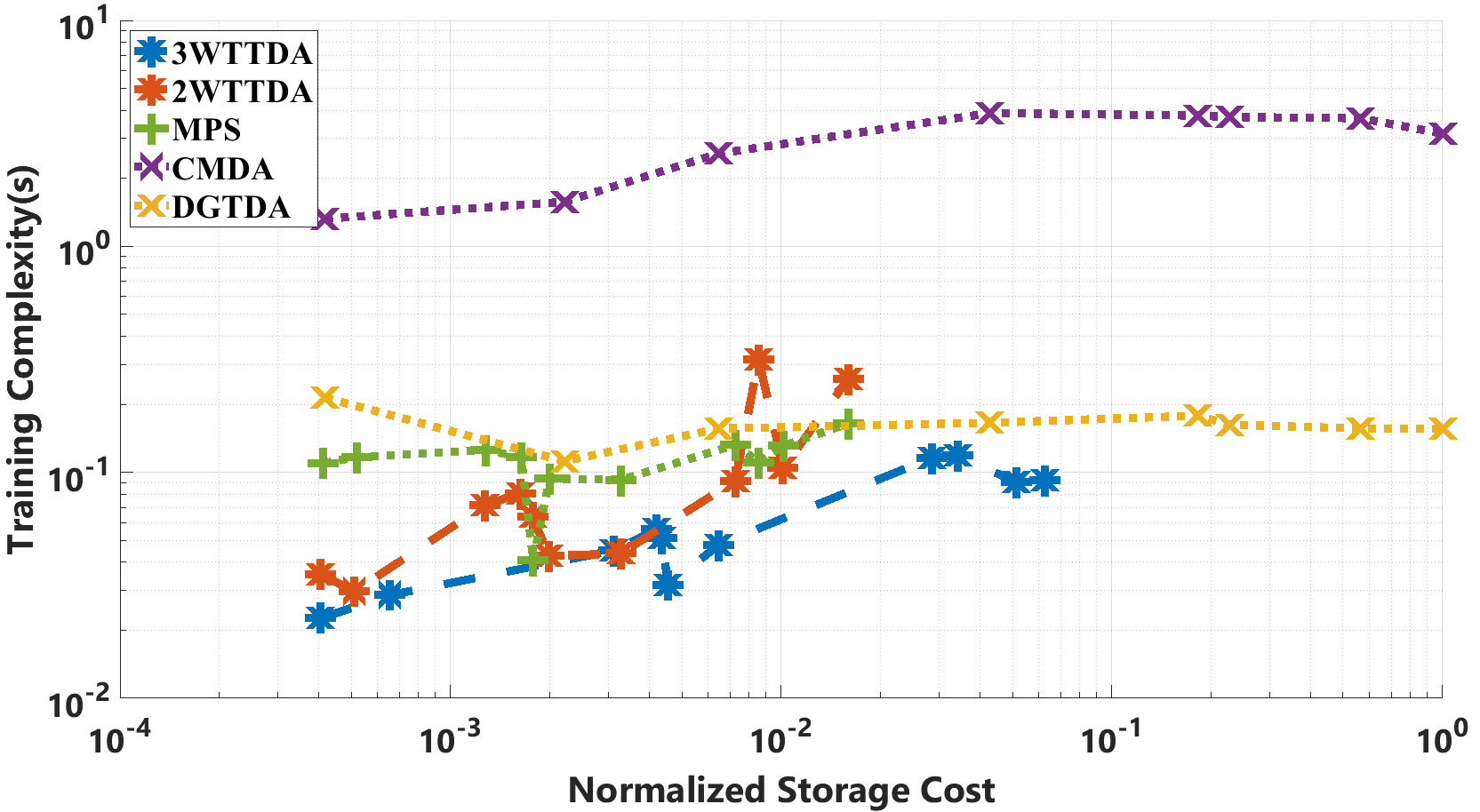}
        \caption{}
        \label{fig:weizsubtime}
    \end{subfigure}
    \begin{subfigure}[b]{.98\columnwidth}
        \centering
        \includegraphics[width=.98\columnwidth]{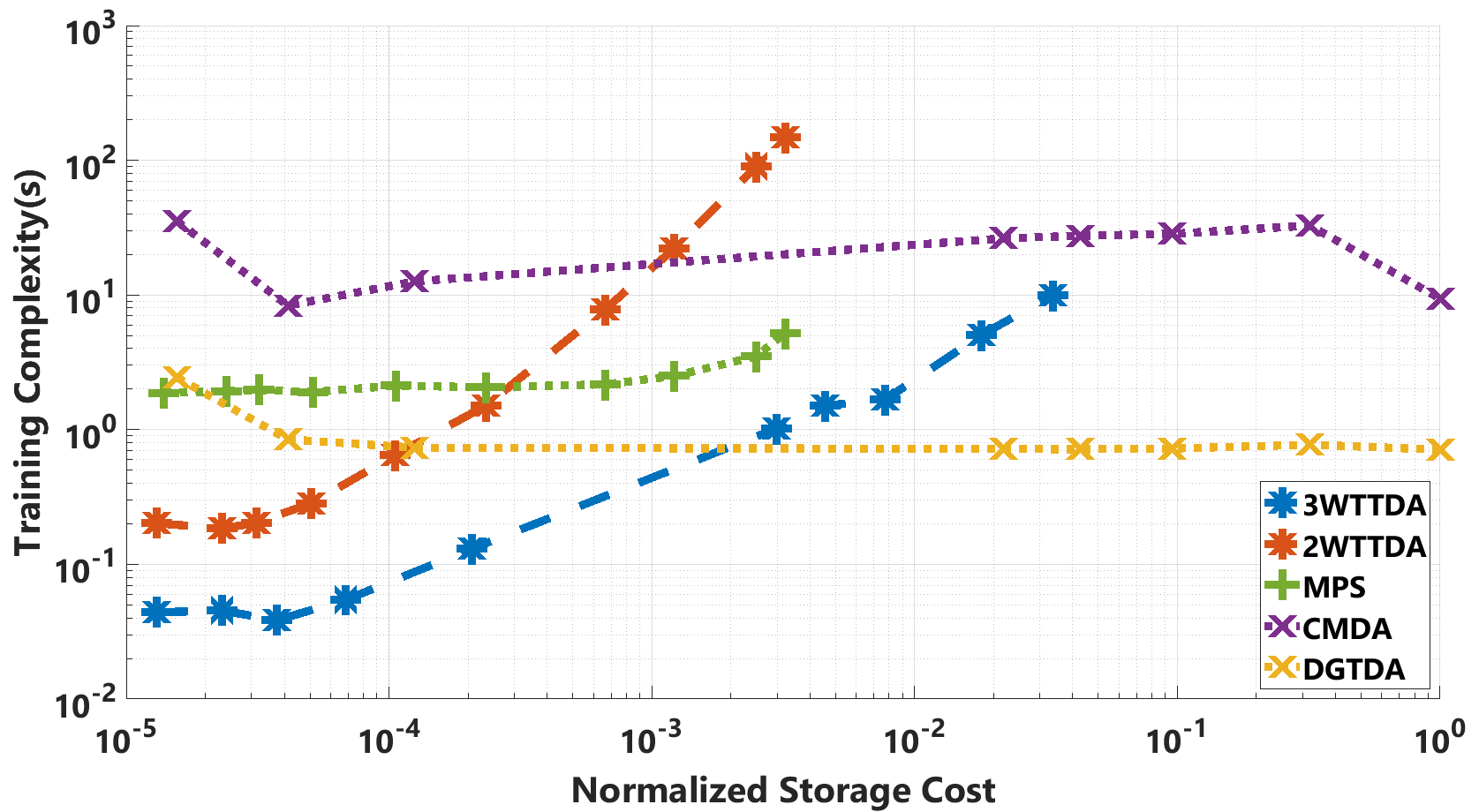}
        \caption{}
        \label{fig:cambsubtime}
    \end{subfigure}
    \begin{subfigure}[b]{.98\columnwidth}
        \centering
        \includegraphics[width=.98\columnwidth]{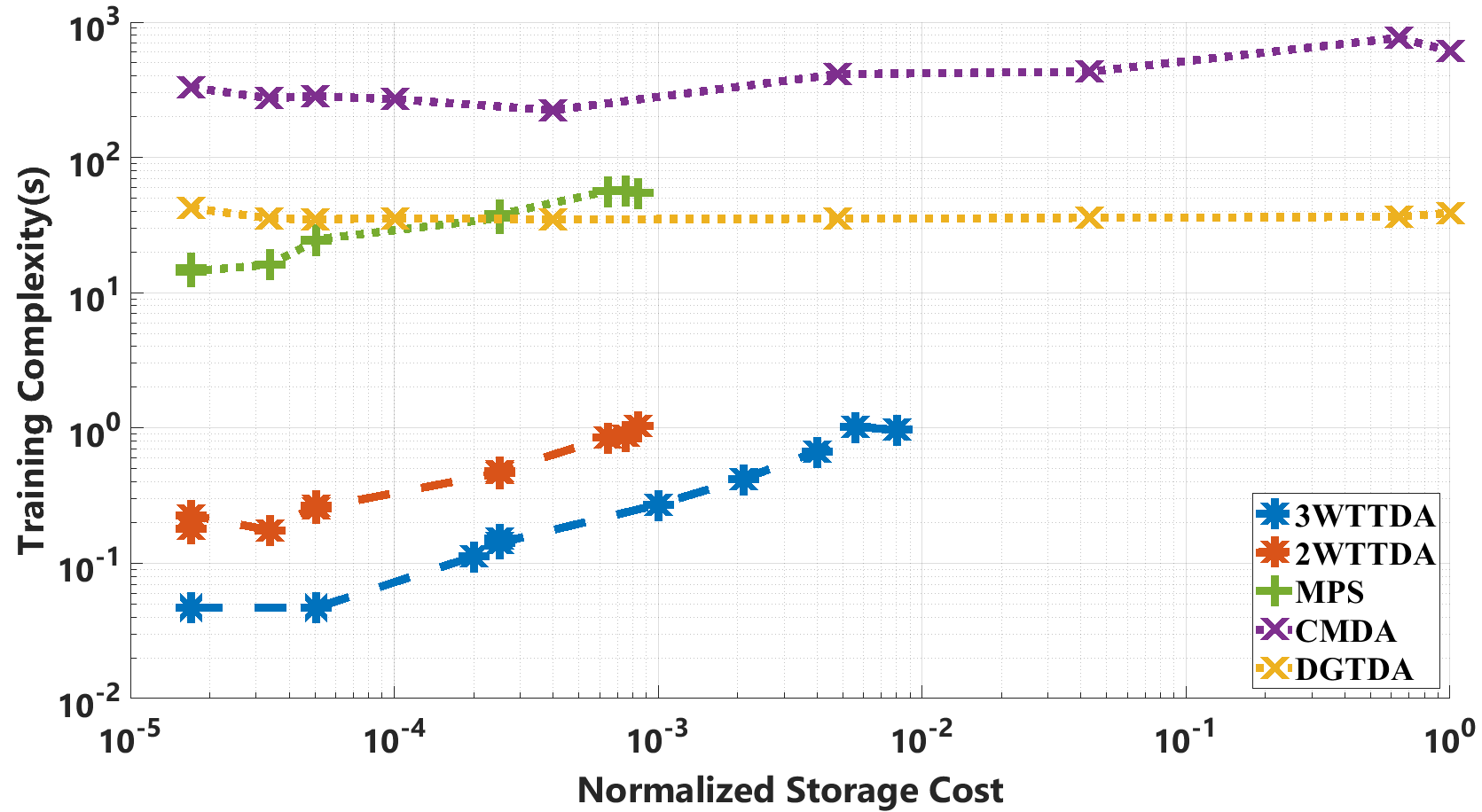}
        \caption{}
        \label{fig:ucfsubtime}
    \end{subfigure}
    \caption{\textcolor{black}{Training complexity  vs.  Normalized storage cost of the different methods for: a) COIL-100, b) Weizmann Face, c) Cambridge Hand Gesture, and d) UCF-101.}}
    \label{fig:subtimeall}
\end{figure*}

\subsection{Classification Accuracy}
We first evaluate the classification accuracy of the different methods with respect to normalized storage complexity. The varying levels of storage cost are obtained by varying the ranks, $R_{i}$s, in the implementation of the tensor decomposition methods.  Varying the truncation parameter $\tau \in (0,1]$, the singular values smaller than $\tau$ times the largest singular value are eliminated. The remaining singular values are used to determine the ranks $R_i$s for both TT-based and TD-based methods. For TT-based methods, the ranks are selected using TT-decomposition proposed in \cite{oseledets2011tensor}, while for TD-based methods truncated HOSVD was used. 

Fig. \ref{fig:clsacc} illustrates the classification accuracy of the different methods with respect to normalized storage complexity for COIL-100 data set. For this particular dataset, we implemented all of the methods mentioned above. It can be seen that the proposed discriminant analysis framework in its original form, TTDA, gives the highest accuracy results followed by TTNPE. However, these two methods only operate at very low storage complexities since the TT-ranks of tensor factors are constrained to be smaller than the corresponding mode's input dimensions. We also implemented STTM, which does not provide compression rates similar to other TT-based methods. This is due to the fact that STTM needs to learn $\frac{C(C-1)}{2}$ classifiers with TT structure. Moreover, these methods have very high computational complexity as will be shown in Section \ref{subsec:trcomp}. For this reason, they will not be included in the comparisons for the other datasets. For a wide range of storage complexities, MPS and 2WTTDA perform the best and have similar accuracy. It can also be seen that the storage costs of MPS and 2WTTDA stop increasing after some point due to rank constraints. This is in line with the theoretical storage complexity analysis presented in Section \ref{sec:compC}. Tucker based methods, such as CMDA and DGTDA, along with the original vector based LDA have lower classification accuracy. 

Fig. \ref{fig:weizclsacc} similarly illustrates the classification accuracy of the different methods on the Weizmann Face Database. For all storage complexities, the proposed 2WTTDA and 3WTTDA perform better than the other methods, including TT based methods such as MPS. 

Fig. \ref{fig:cambclsacc} illustrates the classification accuracy for the Cambridge hand gesture database. In this case, 3WTTDA performs the best for most storage costs. As the number of samples for training, validation and testing is very low for Cambridge dataset, the classification accuracy fluctuates with respect to the dimensionality of the features at normalized storage cost of $0.02$. Similar fluctuations can also be seen in the results of \cite{wang2018tensor}.

Finally, we tested the proposed methods on a more realistic, large sized dataset, UCF-101. For this dataset, TT-based methods perform better than the Tucker based methods. In particular, 2WTTDA performs very close to MPS at low storage costs, whereas 3WTTDA performs well for a range of normalized storage costs and provides the highest accuracy overall.

Even though our methods outperform MPS for most datasets, the classification accuracies get close for UCF-101 and COIL-100. This is due to  the high number of classes in these datasets. As the number of classes increases, the number of scatter matrices that needs to be estimated also increases which results in a larger bias given limited number of training samples. This improved performance of MPS for datasets with large number of classes is also observed when MPS is compared to CMDA. Therefore, the reason that  MPS and the proposed methods perform similarly is a limitation of discriminant analysis rather than the proposed tensor network structure.

\subsection{Training Complexity}
\label{subsec:trcomp}
In order to compute the training complexity, for TT-based methods, each set of tensor factors is optimized until the change in the normalized difference between consecutive tensor factors is less than $0.1$ or $200$ iterations is completed. After updating the factors in a branch, no further optimizations are done on that branch in each iteration. CMDA iteratively optimizes the subspaces for a given number of iterations (which is set to 20 to increase the speed in our experiments) or until the change in the normalized difference between consecutive subspaces is less than $0.1$.

Figs. \ref{fig:subtime}, \ref{fig:weizsubtime}, \ref{fig:cambsubtime}, \ref{fig:ucfsubtime} illustrate the training complexity of the different methods with respect to normalized storage cost for the four different datasets. In particular, Fig. \ref{fig:subtime} illustrates the training complexity of all the methods including TTNPE, TTDA and STTM for COIL-100. It can be seen that STTM has the highest computational complexity among all of the tested methods. This is due to the fact that for a 100-class classification problem, STTM implements $(100)(99)/2$ one vs. one binary classifiers, increasing the computational complexity. Similarly, TTNPE has high computational complexity as it tries to learn the manifold projections which involves eigendecomposition of the embedded graph. Among the remaining methods, LDA has the highest computational complexity as it is based on learning from vectorized samples which increases the dimensionality of the covariance matrices. For the tensor based methods, the proposed 2WTTDA and 3WTTDA have the lowest computational complexity followed by MPS and DGTDA. In particular, for large datasets like UCF-101 the difference in computational complexity between our methods and existing TT-based methods such as MPS is more than a factor of $10^2$. 

\begin{table*}[htb]
    \centering
    \caption{\textcolor{black}{Average classification accuracy (left) and training time (right) with standard deviation for various methods and datasets.}}
    \begin{tabular}{l|c|c|c|c||c|c|c|c|c}
        \hline
        Methods & COIL-100 & Weizmann & Cambridge &UCF-101 &  (s) & COIL-100 & Weizmann & Cambridge & UCF-101\\
        \hline
        3WTTDA &$\mathbf{ 95.6 \pm 0.4}$ & $93.6\pm 2$&$\mathbf{98.2\pm1.7}$ &$\mathbf{68.6\pm0.8}$ & & $\mathbf{0.09 \pm 0.005}$ & $\mathbf{0.05 \pm 0.02}$ &$\mathbf{ 0.11 \pm 0.01}$ &$\mathbf{0.67\pm0.02}$ \\
        2WTTDA & $94.8 \pm 0.5$ & $\mathbf{97.6 \pm 1.2}$& $89.1\pm16.7$ &$67.7\pm0.9$ & & $0.24\pm 0.06$ & $0.09 \pm 0.02$ & $1.7 \pm 1.5$ &$0.853\pm 0.13$\\
        MPS & $94.2 \pm 0.2$ & $87.5 \pm 2.3$& $56.2\pm9.8$ &${67.9\pm0.6}$ & & $1.4\pm 0.13$ & $0.13 \pm 0.01$ & $2.07 \pm 0.25$ &$56.4\pm 1.9$\\
        CMDA & $86.3 \pm 0.7$ & $96.4 \pm 1.03$& $95\pm 2.8$ &$67.7\pm0.8$ & & $12.2\pm6.6$ & $2.6 \pm 0.3$ & $12.6 \pm 0.3$ & $413.5\pm24.1$\\
        DGTDA & $76.6 \pm 0.9$ & $69.9 \pm 1.8$& $35.4\pm8.7$ &$57.3\pm2.7$ & & $0.7\pm0.06$ & $0.16 \pm 0.02$ & $0.7 \pm 0.04$ & $35.3\pm2.9$ 
    \end{tabular}
    \label{tab:acc}
\end{table*}
\subsection{Convergence}
In this section, we present an empirical study of convergence for TTDA in Fig. \ref{fig:conv} where we report the objective value of TTDA, i.e. the expression inside argmin operator in (\ref{eq:TTDAU}), with random initialization of projection tensors. This figure illustrates the convergence of the TTDA algorithm, which is at the core of both 2WTTDA and 3WTTDA, on COIL-100 dataset. It can be seen that even for random initializations of the tensor factors, the algorithm converges in a small number of steps. The convergence rates for 2WTTDA and 3WTTDA are faster than that of TTDA as they update smaller sized projection tensors as shown in Section \ref{sec:compC}.

\begin{figure}
\centering
    \includegraphics[width=.98\columnwidth]{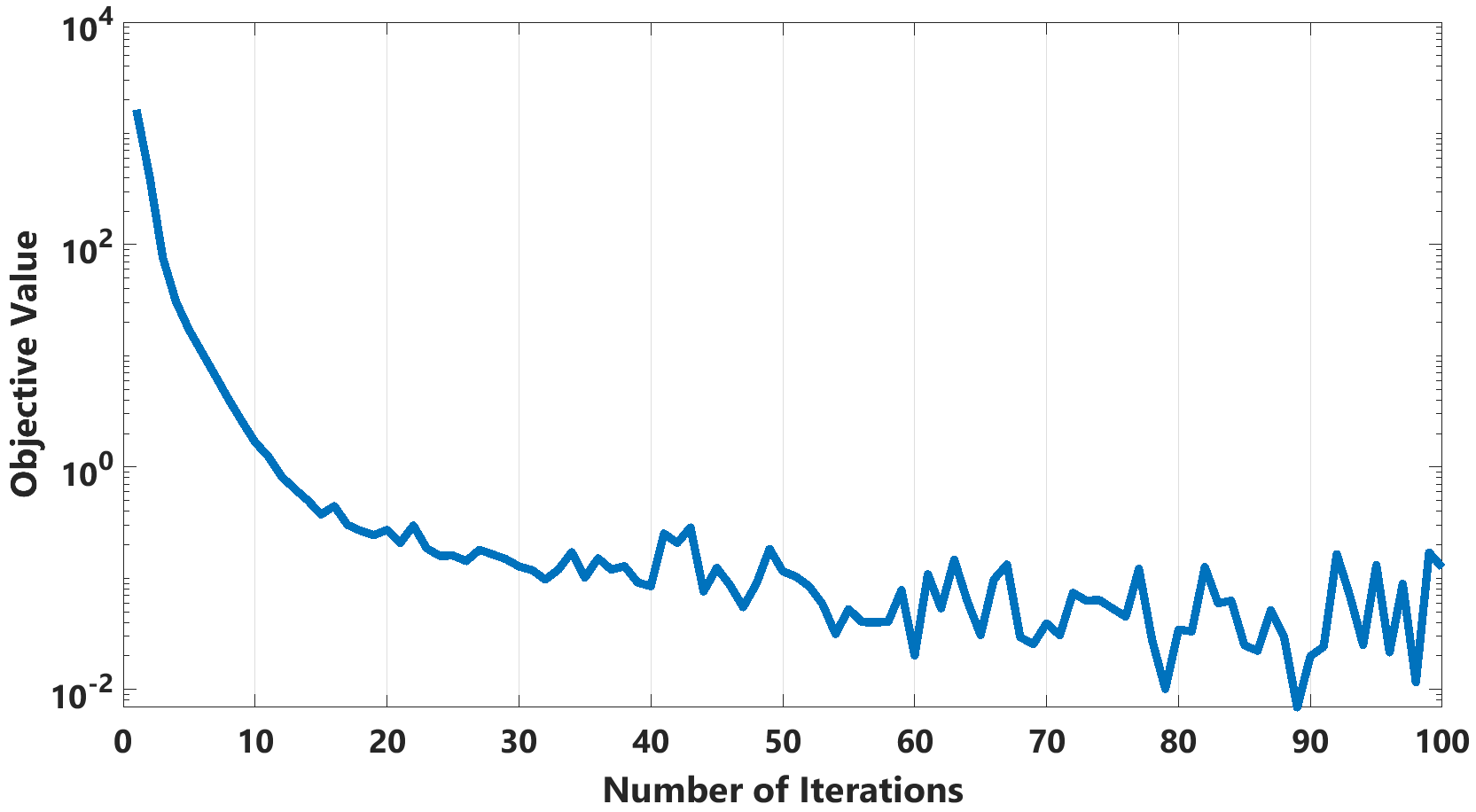}
    \caption{ \textcolor{black}{Convergence curve for TTDA on COIL-100. Objective value vs. the number of iterations is shown. }}
    \label{fig:conv}
\end{figure}

\subsection{Effect of Sample Size on Accuracy}
We also evaluated the effect of training sample size on classification accuracy for Weizmann Dataset. In Fig. \ref{fig:weizsample}, we illustrate the classification accuracy with respect to training sample size for different methods. It can be seen that 3WTTDA provides a high classification accuracy even for small training datasets, i.e., for 15 training samples it provides an accuracy of $96\%$. This is followed by CMDA and 2WTTDA. It should also be noted that DGTDA is the most sensitive to sample size as it cannot even achieve the local optima and more data allows it to learn better classifiers.
\begin{figure}
    \centering
    \includegraphics[width=.98\columnwidth]{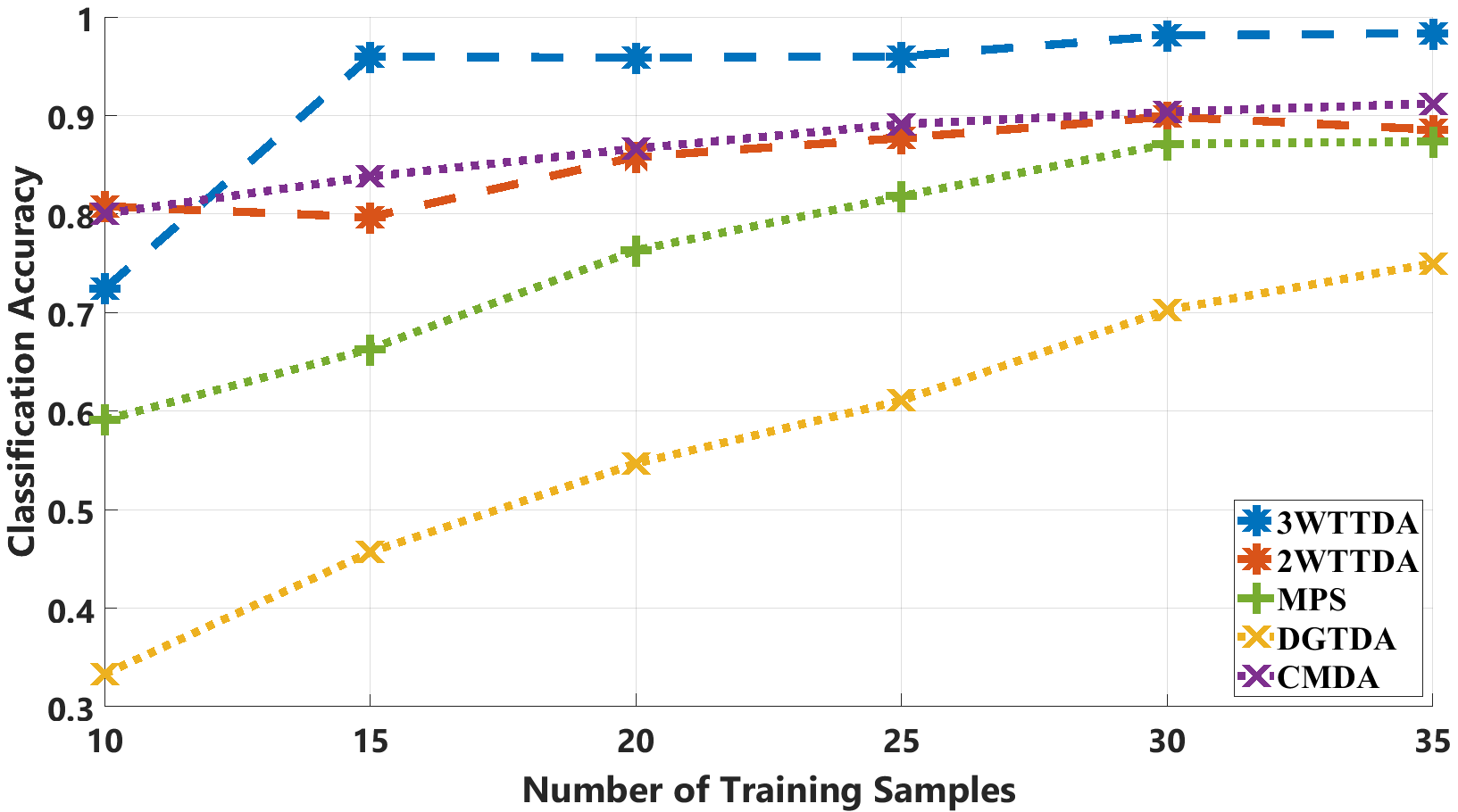}
    \caption{ \textcolor{black}{Comparison of classification accuracy vs. training sample size for Weizmann Face Dataset for different methods.} }
    \label{fig:weizsample}
\end{figure}

\subsection{Summary of Experimental Results}
\label{subsec:summary}
In Table \ref{tab:acc}, we summarize the performance of the different algorithms for the four different datasets considered in this paper. In the left half of this table, we report the classification accuracy (mean $\pm$ std) of the different methods for a fixed normalized storage cost of about $2.10^{-2}$ for COIL-100, $6.10^{-3}$ for Weizmann Face, $2.10^{-4}$ for Cambridge Hand Gesture and $10^{-3}$ for UCF-101 datasets. At the given compression rates, for all datasets the proposed 3WTTDA and 2WTTDA perform better than the other tensor based methods. In some cases, the improvement in classification accuracy is significant, e.g. for Weizmann and Cambridge data sets. These results show that the proposed method achieves the best trade-off, i.e. between normalized storage complexity and classification accuracy.

Similarly, the right half of Table \ref{tab:acc} summarizes the average training complexity for the different methods for the same normalized storage cost. From this Table, it can be seen that 3WTTDA is the most computationally efficient method for all datasets. This is followed by 2WTTDA. The difference in computational time becomes more significant as the size of the dataset increases, e.g. for UCF-101. Therefore, even if the other methods perform well for some of the datasets, the proposed methods provide higher accuracy at a computational complexity reduced by a factor of $10^{2}$.


\section{Conclusions}
In this paper, we proposed a novel approach for tensor-train based discriminant analysis for tensor object classification. The proposed approach first formulated linear discriminant analysis such that the learnt subspaces have a TT structure. The resulting framework, TTDA, reduces storage complexity at the expense of high computational complexity. This increase in computational complexity is then addressed by reshaping the  projection vector into matrices and third-order tensors, resulting in 2WTTDA and 3WTTDA, respectively. A theoretical analysis of storage and computational complexity illustrated the tradeoff between these two quantities and suggest a way to select the optimal number of modes in the reshaping of TT structure. The proposed methods were compared with the state-of-the-art TT-based subspace learning methods as well as tensor based discriminant analysis for four datasets. While providing reduced storage and computational costs, the proposed methods also yield higher or similar classification accuracy compared to state-of-the art tensor based learning methods such as CMDA, STTM, TTNPE and MPS.


 The proposed multi-branch structure can also be extended to unsupervised methods such as dictionary learning, and subspace learning applications. The structure can also be optimized by permuting the modes in a way that the dimensions are better balanced than the original order.


\bibliographystyle{IEEEtran}
\bibliography{ref}




\end{document}